\newif\ifanonymous
\anonymousfalse

\documentclass[11pt,letterpaper]{article}

\usepackage{url}
\usepackage{fullpage}

\usepackage[subtle,mathspacing=normal]{savetrees}
\usepackage{tikz}
\usetikzlibrary{decorations.pathreplacing}
\usepackage{graphicx}

\usepackage{amsmath}
\usepackage{algorithm}
\usepackage[noend]{algpseudocode}
\usepackage{amssymb}
\usepackage{amsthm} %
\usepackage{color}
\usepackage{array}
\usepackage{xy}
\usepackage{setspace}
\usepackage{varwidth}%
\usepackage{multicol}
\usepackage{algorithm}
\usepackage{hyperref}
\providecommand{\email}[1]{\href{mailto:#1}{#1}} %
\usepackage{thm-restate}
\usepackage[disable]{todonotes}

\newcommand{\poly}{\operatorname{poly}}

 \newcommand{\E}{\mathbb{E}}

\newtheorem{thm}{Theorem}[section]

\theoremstyle{remark}

\newtheorem{theorem}{Theorem}[section]

\newtheorem{lemma}[thm]{Lemma}

\newtheorem{corollary}[thm]{Corollary}

\theoremstyle{remark}

\newcommand{\defn}[1]{\textbf{\emph{#1}}}
\renewcommand{\paragraph}[1]{\vspace{.2 cm} \noindent \textbf{#1}}

\renewcommand{\epsilon}{\varepsilon}
\renewcommand{\emptyset}{\varnothing}

\newcommand{\clip}{\operatorname{clip}}

\begin{document}
\title{Optimal Worst-Case Open Addressing: Braided Cuckoo Hashing}
\hypersetup{pdftitle={Optimal Worst-Case Open Addressing: Braided Cuckoo Hashing}}
\ifanonymous
\author{Anonymous}
\hypersetup{pdfauthor={}}
\else
\author{Andrew Krapivin\thanks{Carnegie Mellon University (\email{andrew@krapivin.net}).}
\and William Kuszmaul\thanks{Carnegie Mellon University (\email{kuszmaul@cmu.edu}).}}
\hypersetup{pdfauthor={Andrew Krapivin, William Kuszmaul}}
\fi
\date{}
\maketitle

\begin{abstract}
    In this paper, we introduce braided cuckoo hashing, an open-addressed hash table that supports load factor up to $1 - \epsilon$ (for some parameter $\epsilon \in (0, 1)$), while supporting queries, insertions, and deletions in \emph{worst-case} time $O(\log \epsilon^{-1})$, with high probability. This result represents an exponential improvement on insertion/deletion time over previous constructions, and disproves a conjecture of Bender, Kuszmaul, and Zhou that the product of worst-case query time and expected insertion/deletion time must be $\Omega(\epsilon^{-1})$.
\end{abstract}

\section{Introduction}
First popularized in the early 2000s \cite{PR01,FPSS05,DW07}, \defn{cuckoo hashing} is a method of implementing an open-addressed hash table with \emph{worst-case} time bounds for queries. In its simplest form, cuckoo hashing maintains the invariant that each element $x$ is in one of two random slots $h_1(x) \in [1, m/2]$ and $h_2(x) \in (m/2, m]$ in an array of size $m$. To query whether an element $x$ is present, one simply checks the slots $h_1(x)$ and $h_2(x)$. To insert an element $x$, one performs a \defn{kickout chain}, placing $x$ in position $h_1(x)$, evicting any element $x'$ in position $h_1(x)$ to its other viable slot $h_2(x')$, evicting any element $x''$ in position $h_2(x')$ to its other viable slot $h_1(x'')$, and so on, until the chain terminates at a free slot. Finally, to delete an element $x$, one simply removes it from the table.

It is not clear, \emph{a priori}, whether insertions in a cuckoo hash table should even succeed. (The eviction chain might enter an infinite loop.) However, in a celebrated result of Pagh and Rodler \cite{PR01}, it was shown that, so long as the hash table is less than a $1/2 - \Omega(1)$ fraction full, each insertion succeeds with probability $1 - O(1/n^2)$, and the expected insertion time is $O(1)$.

To support higher load factors, researchers have considered extensions of cuckoo hashing in which each query checks $d > 2$ slots $h_1(x), \ldots, h_d(x)$ in the hash table. These slots can be selected in a variety of ways, for example by using $d$ independent hash functions $h_1, \ldots, h_d$ \cite{FPSS05}, by breaking slots into buckets and hashing to buckets whose sizes sum to $d$ \cite{DW07}, or by other more elaborate schemes \cite{FPSS05,deamortizedcuckoo,stashcuckoo,bender2023iceberg,walzer2023load}. To support load factor $1 - \epsilon$, the optimal value of $d$, and indeed the optimal worst-case query time in any open-addressed hash table, is known to be $\Theta(\log \epsilon^{-1})$. 

This raises a natural question: If we wish to support worst-case query time $\Theta(\log \epsilon^{-1})$, what is the best insertion/deletion time that we can hope for? Early work on this question \cite{FPSS05,DW07} led to bounds which were super-polynomial in $\epsilon^{-1}$ (and in some cases, a function of $n$ \cite{PolylogAlan,FPS13}), but conjectured that a bound of the form $O(\epsilon^{-1})$ should be possible \cite{FPSS05,DW07}. This was finally achieved in 2025 by Kuszmaul and Mitzenmacher \cite{KM25}, who showed how to achieve an expected insertion/deletion time of $O(\epsilon^{-1})$.

There is good reason to believe, intuitively, that a bound of $\tilde{O}(\epsilon^{-1})$ should be optimal. Consider, in particular, a long-running workload that alternates between insertions and deletions at load factor $1 - \epsilon$. To a first approximation, each deletion creates a ``hole'' in a uniformly random slot. Insertions then need to ``track down'' these holes in order to fill them. But each time an insertion probes $h_i(x)$ for some element $x$ and position $i$, it has probability at most $O(\epsilon)$ of finding such a hole. Even if the element can somehow check all of $h_1(x), h_2(x), \ldots, h_d(x)$ at once (for example, if they are arranged in a bucket), this still indicates a time bound of $\Omega(\epsilon^{-1} / d) = \Omega(\epsilon^{-1} / \log \epsilon^{-1})$. In general, Bender, Kuszmaul, and Zhou \cite{bender2024tight} conjectured that the product of worst-case query time and expected insertion/deletion time must be $\Omega(\epsilon^{-1})$ for any open-addressed hash table.

\paragraph{This paper: Optimal worst-case open addressing. }We present an open-addressed hash table, called \defn{braided cuckoo hashing}, that supports worst-case time $\Theta(\log \epsilon^{-1})$ for \emph{all of} queries, insertions, and deletions, with high probability. This represents an exponential improvement on insertion/deletion time over the bound of $\tilde{\Theta}(\epsilon^{-1})$ that was previously conjectured to be optimal.

\begin{restatable}[Main Theorem]{theorem}{maintheorem}
    \label{thm:main}
    Let $n^{-1/4} \le \epsilon \le 1/2$.
    There is an open-addressed hash table using $(1 + \epsilon)n$ slots and
    $O(\log n)$ machine words of metadata, that offers the following
    guarantee for any oblivious polynomial-length insertion/deletion sequence in which the number of live keys
    never exceeds $n$: With high probability in $n$, every operation succeeds and every
    insertion and deletion takes worst-case time $O(\log\epsilon^{-1})$; additionally, queries take worst-case time $O(\log\epsilon^{-1})$.
\end{restatable}

The simplest parameter regime to understand braided cuckoo hashing in is the case where $\epsilon = n^{-\Theta(1)}$. In this case, each element $x$ hashes to two bins $h_1(x)$ and $h_2(x)$ that are each of size $\Theta(\log n)$. But, unlike in traditional bucketized cuckoo hashing \cite{DW07}, these bins are not independent (or even close to independent). Instead, $h_2$ is chosen in such a way that, conditioned on $h_1(x)$, there are only $O(1)$ options for $h_2(x)$. 

Intuitively, such a strong relationship between $h_1$ and $h_2$ would seem to defeat the point of having two hash functions to begin with, potentially making it impossible to place $n$ elements into the hash table. However, by relating all of the bins in the hash table to each other in a specific way (see the braided structure in Figure \ref{fig:four-layer-construction}), we are able to preserve the ``power-of-two-choice'' type phenomenon that makes bucketized cuckoo hashing possible, even though the function $h_2$ that we use has barely any additional entropy over $h_1$. This is the main technical insight that makes the paper possible.

What does the relationship between $h_1$ and $h_2$ buy us? It enables deletions to perform what we call \defn{pullback chains}, which are essentially the opposite of eviction chains: instead of evicting an element from the current bin to another bin, we fill a free slot in the current bin by \emph{obtaining} an element from some other bin.
In normal bucketized cuckoo hashing, pullback chains are not possible, because the elements $y$ that are capable of residing in a given bin $h_1(x)$ could be almost anywhere in the hash table. But, in braided cuckoo hashing, the elements $y$ capable of residing in a given bin $h_1(x)$ are themselves in one of only $O(1)$ bins! This is what makes pullback chains possible. 

Using pullback chains, we can maintain the invariant that (if the hash table is almost full) almost all the free slots are in some specific small ``layer'' of the hash table (the bottom layer in Figure \ref{fig:four-layer-construction}). Insertions can then reach these free slots by performing an eviction chain that targets that layer, instead of having to look for them uniformly at random throughout the hash table. 

To make this high-level structure work, there are a number of additional technical ingredients that we end up needing. To prove (even in the static setting) that the placement of elements is possible, we must analyze a certain balls-and-bins process (Theorem \ref{thm:static-placement}). This is where the ``braided structure'' of the hash table ends up being important, as it allows us to leverage properties of sub-Gaussian random variables to show that the process is well-behaved. 

Additionally, to avoid certain types of problematic dependencies, we design our basic data structure (in the parameter regime $\epsilon = n^{-\Theta(1)}$) to be \emph{history independent}, meaning that the current state of the hash table depends only on the current set of elements, and not in any other way on the history of past insertions and deletions. History independence has historically been viewed as a \emph{privacy guarantee} \cite{antipers}, but in recent years has emerged as a useful technical tool for designing dynamic data structures \cite{kuszmaul2023strongly,bender2022online,bansal2022balanced,bender2026history,chechik2019fully,behnezhad2019fully,das2026history}. Our results can be viewed as a (slightly unexpected) extension of that trend. 

Finally, to extend our construction to larger values of $\epsilon$ (not $\epsilon = n^{-\Theta(1)}$), we must extend the construction so that, when certain failure events occur, elements still have ``backup'' options for where they can go. This is achieved using a backyard \cite{arbitman2010backyard}, which is a small secondary hash table that absorbs elements with nowhere else to go.

When we put all the pieces together, we are able to prove Theorem \ref{thm:main}. Additionally, we also prove a lower bound on expected insertion/deletion time. We show, via a simple argument about one-way communication protocols, that any open-addressed hash table with worst-case $O(\log \epsilon^{-1})$-time queries must incur $\Omega((\log \epsilon^{-1}) / \log \log \epsilon^{-1})$ expected-time insertions/deletions. This means that our upper bound on insertion/deletion time is tight up to an $O(\log \log \epsilon^{-1})$ factor.

\paragraph{Related work. }
The concept of an \emph{open-addressed hash table}, which is a hash table in which queries to an element $x$ search through a fixed sequence $h_1(x), h_2(x), \ldots \in [m]$ of positions to find the element being queried, dates back to the early 1950s with the introduction of linear probing \cite{peterson1957addressing,knuth1963notes,knuthvol3} and random-probing \cite{knuthvol3} hash tables. The first open-addressed hash table to achieve \emph{worst-case} time bounds was the cuckoo hash table \cite{PR01,lyon1985achieving,hagersten2001shared}, which Pagh and Rodler \cite{PR01} analyzed at any load factor $1/2 - \Omega(1)$. This prompted a long line of work \cite{FPSS05,PolylogAlan,FPS13,wearmin,FriezeJohansson,newinsertion,Walzer,walzer2023load,bell20241,KM25,bell2026fast} on variations of cuckoo hashing, including $d$-ary cuckoo hashing \cite{FPSS05} and  bucketized cuckoo hashing \cite{DW07} that seek to achieve higher load factors. The main goal of this line of work was to achieve the optimal query time $O(\log \epsilon^{-1})$ with as small an insertion/deletion time as possible. This culminated in a 2025 result of Kuszmaul and Mitzenmacher \cite{KM25}, who showed how to achieve an expected insertion/deletion time of $O(\epsilon^{-1})$ -- a bound which was conjectured to be optimal up to polylogarithmic factors by Bender, Kuszmaul, and Zhou \cite{bender2024tight}. The main result of this paper disproves their conjecture, bringing the insertion/deletion time down to $O(\log \epsilon^{-1})$.

In addition to $d$-ary cuckoo hashing (where the $d$ hash functions are independent) and bucketized cuckoo hashing (where each element hashes to independent bins), there have been a number of other variations of cuckoo hashing that researchers have studied. These include deamortized cuckoo hashing \cite{deamortizedcuckoo,stashcuckoo}, which uses a queueing/stashing idea to convert insertion time bounds from expected to worst-case (we use this result directly in Section \ref{sec:all-load-factors}); backyard cuckoo hashing and iceberg hashing \cite{arbitman2010backyard,bender2023iceberg}, which are succinct hash tables (built on top of open-addressed hash tables) that guarantee each element to be in one of $\poly(\epsilon^{-1})$ positions; and filter hashing \cite{FPSS05}, which is an \emph{insertion-only} variant of cuckoo hashing that achieves insertion and query time bounds of $O(\log^2 \epsilon^{-1})$. Filter hashing is related to braided cuckoo hashing in that both maintain a layered table in which the free slots are kept in the smallest layer -- but the approach in filter hashing works only for insertion workloads (no deletions), since deletions would create free slots in essentially random positions. 

Related to cuckoo hashing is the so-called \emph{power-of-two choices} for load balancing \cite{azar1994balanced,berenbrink2000balanced}, which says that, if one places $m$ balls into $n$ bins, and each ball chooses the emptier of two random bins, then the fullest bin will have only $m/n + O(\log \log n)$ balls. Cuckoo hashing is sometimes also described as a ``power-of-two-choice'' type phenomenon, because it accomplishes something with two hash functions that is impossible to achieve with one.

Open-addressed hash tables (including cuckoo hashing) are typically analyzed under the assumption that the hash functions involved are fully random. There is also a large body of work \cite{Siegel1989FastHighPerformanceHash,PaghPagh,Siegel2004ExtremelyRandomConstantTime,dietzfelbinger2007design,OstlinPagh,dahlgaard2015hashing,thorup2013simple,bercea2023locally} on how to simulate random hash functions efficiently; as well as well-established techniques for applying these hash functions to essentially any hash-table construction (see, e.g., the splitting trick in \cite{dietzfelbinger2009applications}).

Although this paper focuses on worst-case open-addressing, it is also interesting to study hash tables that depart from the open-addressing model, or that are open-addressed but offer non-worst-case guarantees. There have been long lines of work on open addressing with small expected-time operations \cite{brent1973reducing,munro1986techniques,farach2024optimal,farach2026greedy,bender2024tight,azarmehr2026power,kuszmaul2026analysis}, on \emph{non-oblivious} open addressing, where queries may follow an adaptive memory pattern \cite{fiat1992nonoblivious,fiat1993implicit,bender2025optimal}, on succinct and compact dictionaries \cite{raman2003succinct,demaine2006dictionariis,arbitman2010backyard,bender2023iceberg,bender2022optimal,li2023tight,li2024dynamic,blelloch2026optimal}, on stable dictionaries \cite{yao1985uniform,demaine2006dictionariis,bender2023iceberg,bender2025tiny,pandey2023iceberght,tang2026succinct,farach2024optimal,farach2026greedy,krapivin2026linear,zamir2026locality}, on dictionaries with extremely high-probability guarantees \cite{goodrich2012cache,bender2023iceberg,kuszmaul2022hash}, on open addressing in external memory \cite{jensen2008optimality,bender2021linear,braverman2024tight,krapivin2026linear,zamir2026locality}, and on variations of open addressing that are widely used in practice \cite{bell1970linear,guibas1976analysis,lueker1988more,mitzenmacher2012peeling,maurer1968programming,radke1970use,hopgood1972quadratic,kuszmaul2024towards,hu2026revisited,guo2026simple,zamir2026locality,hu2026quadratic}.

\section{Preliminaries}

Throughout the paper, we will be interested in constructing an \emph{open-addressed hash table}, which is a hash table in which queries to a key $x$ visit a predetermined sequence of positions $h_1(x), h_2(x), \ldots, h_d(x) \in [m]$ within some array. We will use $n$ to denote the maximum number of keys ever present at a time, and $m$ to denote the total number of \defn{slots} (also called \defn{cells}) in the hash table. Our constructions will achieve $m = (1 + O(\epsilon))n$ for some load factor parameter $\epsilon \ge 0$.

Our constructions will offer high-probability guarantees. 
Formally, we say that an algorithm depending on sufficiently large constants $c_1, c_2, \ldots$ 
supports a guarantee \defn{with high probability in $n$} if, for any constant $\gamma > 0$, 
there exist sufficiently large values for $c_1, c_2, \ldots$ (still constants) such that
the guarantee holds with probability at least $1 - O(1 / n^{\gamma})$.

Throughout the paper, we will assume fully random hash functions. One could also straightforwardly implement our constructions with explicit hash functions by, say, combining the $n^{\Omega(1)}$-independent hash functions of \cite{Siegel2004ExtremelyRandomConstantTime} with the splitting trick in \cite{dietzfelbinger2009applications}.

Throughout the paper, we will make use of standard Chernoff bounds. Additionally, several proofs will make use of McDiarmid's inequality \cite{mcdiarmid1989method}, a specialization of Azuma's inequality for martingales with bounded differences. As a reference, we state McDiarmid's inequality here.
\begin{lemma}[McDiarmid's Inequality]
\label{lem:mcdiarmid}
Let $X_1, \ldots, X_n$ be independent random variables from some set $S$, and let $f: S^n \to \mathbb{R}$ be a function that is $1$-Lipschitz (changing any input to $f$ changes the output by at most $1$). Then,
\[
    \Pr[|f(X_1, \ldots, X_n) - \E[f(X_1, \ldots, X_n)]| > k \sqrt{n}] \le e^{-\Omega(k^2)}. 
\]
\end{lemma}

\section{Paper Outline}

Most of the rest of the paper is devoted to the proof of Theorem \ref{thm:main}. The high-level structure is as follows.

We begin in Sections \ref{sec:construction}, \ref{sec:static-placement}, and \ref{sec:dynamic-placement}, by focusing on a specific parameter regime that ends up being important for the general construction. Specifically, in these sections, we make two simplifying assumptions: 
\begin{itemize}
\item \textbf{The High-Load Assumption: } The parameter $\epsilon$ is very close to $0$, satisfying $\epsilon = \Theta(n^{-0.25})$. This means that the target time per operation is $\Theta(\log \epsilon^{-1}) = \Theta(\log n)$.
\item \textbf{The Divisibility Assumption:} The maximum number of keys $n$ satisfies a certain specific divisibility condition outlined in Subsection \ref{sec:divisibility-assumption}.
\end{itemize}

With these assumptions in place, we show how to construct a solution that supports $\Theta(\log n)$-time insertions, deletions, and queries. This construction is itself broken into three sections: Section \ref{sec:construction}, which describes the array layout and how to determine which bins a given key hashes to; Section \ref{sec:static-placement}, which considers the \emph{static} setting, in which up to $(1  - \epsilon)m$ keys are given to us in advance; and then Section \ref{sec:dynamic-placement}, which extends the construction to support efficient insertions and deletions. Together, these three sections prove Theorem \ref{thm:main} for the special case where the assumptions above hold.

Then, in Sections \ref{sec:general-capacities}, \ref{sec:all-load-factors}, and \ref{sec:deamortization}, we extend the construction to prove the full version of Theorem \ref{thm:main}. We begin, in Section \ref{sec:general-capacities}, by removing the Divisibility Assumption (this is relatively straightforward). Then, in Section \ref{sec:all-load-factors}, we remove the High-Load Assumption, resulting in a solution that handles arbitrary values of $\epsilon \ge \Omega(n^{-1/4})$ with amortized $O(\log \epsilon^{-1})$ time insertions/deletions and worst-case $O(\log \epsilon^{-1})$ time queries. Finally, in Section \ref{sec:deamortization}, we show how to deamortize the data structure so that insertions and deletions take worst-case $O(\log \epsilon^{-1})$ time (with high probability). This completes the proof of Theorem \ref{thm:main}. 

The most technically involved part of the proof is actually the static setting (Section \ref{sec:static-placement}), where we need to show that the data structure is \emph{even capable} of storing $(1 - \epsilon)m$ keys without irresolvable conflicts. To streamline the presentation, we defer a large part of this proof to Section \ref{sec:static-placement-proof}, which appears at the end of the paper. 

Finally, in addition to proving Theorem \ref{thm:main}, we also, in Section \ref{sec:lower}, prove a lower bound on expected insertion/deletion time: we show, with a simple communication lower-bound based argument, that any open-addressed hash table with worst-case $O(\log \epsilon^{-1})$-time queries must incur $\Omega((\log \epsilon^{-1}) / \log \log \epsilon^{-1})$ expected-time insertions/deletions. Whether or not this $\log \log \epsilon^{-1}$ factor should appear in the lower or upper bound remains open.

\section{The Basic Hash-Table Layout}\label{sec:construction}

In this section, we describe the basic layout of the hash table, including how to decompose the table into bins, and then how to determine which two bins a given key hashes to. 

The construction is parameterized by the maximum number $n$ of keys that will ever be present at a time. Given $n$, we will construct a hash table with $m = n + \epsilon n$ slots, where in this section (and in Sections \ref{sec:static-placement} and \ref{sec:dynamic-placement}), we will support only the specific load factor $\epsilon = \Theta(n^{-0.25})$ (we call this the \defn{High-Load Assumption}). In addition to the parameters $n$ and $m$, the construction will make use of a large positive constant $c$, and use ``bins'' of size $B := \lceil c \log n\rceil$. Finally, the construction will make use of two parameters $K$ and $\ell$ coming from the Divisibility Assumption defined in Subsection \ref{sec:divisibility-assumption} -- these parameters also serve concrete roles in the construction, with $\ell$ serving as the number of ``layers'' in the hash table, and $2^K$ serving as the number of bins in the smallest layer. In addition to these parameters, we will also introduce several other families of parameters during the construction that will be useful as notation throughout the analysis: the layers $L_1, \ldots, L_\ell$, the number $b_i$ of bins in each layer $L_i$, and the number $s_i$ of slots in each layer $L_i$, and the hash functions $h_1$ and $h_2$, and the probability $p_i$ that $h_1$ selects layer $L_i$. A table consisting of all the notations is given in Figure \ref{fig:notations}.

\begin{figure}[t]
    \centering
    \small
    \setlength{\tabcolsep}{4pt}
    \begin{tabular}{@{}l p{0.50\textwidth} p{0.36\textwidth}@{}}
        \hline
        \textbf{Parameter} & \textbf{Meaning} & \textbf{Formula (If Applicable)} \\
        \hline
        \(n\) &
        Maximum number of keys at a time &
        \\
        \(m\) &
        Total number of slots in the hash table &
        \(m = n + \epsilon n\) \\
        \(\epsilon\) &
        Load factor parameter & \\
        \(c\) &
        Large positive constant used to set the bin size &
        \\
        \(B\) &
        Number of slots in each bin &
        \(B = \lceil c \log n\rceil\) \\
        \(K\) &
        Integer parameter from the Divisibility Assumption; \phantom{foo} \(2^K\) is the number of bins in the final layer &
        \(K = \frac34\log_2 n - O(\log\log n)\) \\
        \(\ell\) &
        Number of layers in the hash table &
        \(\ell = \frac14\log_2 n + O(1)\) \\
        \(L_i\) &
        Layer \(i\) of the hash table, for \(i=1,\ldots,\ell\) &
        \\
        \(b_i\) &
        Number of bins in layer \(L_i\) &
        \(b_i = 2^{K+\ell-i}\) \\
        \(s_i\) &
        Number of slots in layer \(L_i\) &
        \(s_i = B b_i = B2^{K+\ell-i}\) \\
        \(h_1, h_2\) &
        First and second hash functions mapping a key to a bin &
        \\
        \(p_i\) &
        Probability that \(h_1\) selects layer \(L_i\) &
        $n \sum_{j<i}p_j = 
    \sum_{j<i}s_j+\frac12s_i$ \\
        \hline
    \end{tabular}
    \caption{Summary of notation for the basic hash-table layout.}
    \label{fig:notations}
\end{figure}

\subsection{The Divisibility Assumption}\label{sec:divisibility-assumption}

In addition to the High-Load Assumption, we will also make the following divisibility assumption on our parameters. 
Given the parameter $n$ (the maximum number of keys) and the parameter $B$ (the bin size, given by $\lceil c \log n\rceil$), 
we will assume there exist integers \(K\) and \(\ell\) satisfying
\[
    n = B2^K\left(2^\ell-\frac32\right),
\]
with
\[
    \ell = \frac14\log_2 n \pm O(1)
    \qquad\text{and}\qquad
    K = \frac34\log_2 n - \log B \pm O(1).
\]
Equivalently,
\[
    2n = B2^K(2^{\ell+1}-3).
\]
This assumption, which we call the \defn{Divisibility Assumption}, is used to make the layer sizes in our construction exact powers of
two.  The assumption will later be removed in Section \ref{sec:general-capacities}.

\subsection{Braided Cuckoo Hashing}

With both the High-Load Assumption and the Divisibility Assumption in mind, we now present the 
basic layout of a braided cuckoo hash table.

\paragraph{Layers and bins.}
The hash table is organized into layers \(L_1,L_2,\ldots,L_\ell\), where each layer
has a power-of-two number of bins that each contain exactly \(B\) slots. For \(i=1,\ldots,\ell\), define $$b_i =
2^{K+\ell-i}$$ to be the number of bins in layer \(L_i\), and define \(s_i = B
b_i\) to be the number of slots in layer \(L_i\). Thus \(s_i =
B2^{K+\ell-i}\).

As we check below, the total number of slots in the table is $m = n+\frac12s_\ell = n + \Theta(n^{3/4}),$
and thus $\epsilon = \Theta(n^{-0.25})$.

\begin{lemma}
    The total number of slots $m$ in the table satisfies $m = n+\frac12s_\ell = n + \Theta(n^{3/4})$.
\end{lemma}
\begin{proof}
The total number of slots before the final layer is
\[
    \sum_{i<\ell}s_i
    =
    B2^K\sum_{i<\ell}2^{\ell-i}
    =
    B2^K(2^\ell-2).
\]
The final layer has \(s_\ell = B2^K\) slots. Hence
\[
    \sum_{i<\ell}s_i+\frac12s_\ell
    =
    B2^K\left(2^\ell-\frac32\right)
    =
    n.
\]
By the assumptions on \(K\), the final layer has \(s_\ell = B2^K =
\Theta(n^{3/4})\), meaning that $$m = n+\Theta(n^{3/4}),$$
as desired.
\end{proof}

One should think of the final layer \(L_\ell\) serving as a slack layer. When there are $n$ elements in the hash table, they will (with high probability) be organized in a way so that all the previous layers are 100\% full, while the final layer is only 50\% full.

\paragraph{The first hash function. } Within the hash table, each element $x$ hashes to two bins $h_1(x)$ and $h_2(x)$.

The hash function \(h_1(x)\) first selects a layer, where each layer is selected with some probability $p_i$ to be defined in a moment, and then selects a uniformly random bin within that layer.

The final layer is never selected by \(h_1\), so $p_\ell = 0$. 
For every \(2\le i\le \ell\), the probabilities are chosen so that
\[
    n\sum_{j<i}p_j
    =
    \sum_{j<i}s_j+\frac12s_i.
\]
Equivalently,
\[
    \sum_{j<i}p_j
    =
    \frac1n
    \left(
        \sum_{j<i}s_j+\frac12s_i
    \right).
\]

In words, before layer \(L_i\), the expected number of elements whose first hash
lies in earlier layers is equal to the number of slots in those earlier layers,
plus half the number of slots in layer \(L_i\).

The equation for \(i=\ell\) gives
\[
n\sum_{j<\ell}p_j
    =
    \sum_{j<\ell}s_j+\frac12s_\ell
    =
    n,
\]
as in the proof of the lemma above,
and hence $\sum_{j<\ell}p_j=1$.
Thus \(p_1,\ldots,p_{\ell-1}\) form a probability distribution over the
non-final layers (while \(p_\ell=0\)).

\paragraph{The second hash function. } We now define \(h_2(x)\). Suppose that \(h_1(x)\) selects a bin in layer \(L_i\),
where \(i<\ell\). Since
\[
    b_i = 2^{K+\ell-i},
\]
we may identify the bins of \(L_i\) with binary strings of length
\[
    k_i = K+\ell-i.
\]
Write the bin selected by \(h_1(x)\) as
\[
    x_1x_2\cdots x_{k_i}\in\{0,1\}^{k_i}.
\]
The second hash \(h_2(x)\) is placed in layer \(L_{i+1}\). Since $b_{i+1}=\frac{b_i}{2}$, the bins of \(L_{i+1}\) may be identified with binary strings of length $k_i-1$.

To obtain the bin in $L_{i+1}$ selected by \(h_2(x)\), first delete the final bit
\(x_{k_i}\), leaving $x_1x_2\cdots x_{k_i-1}$.
Next, consider the pair of bits in positions $2i-1$ and $2i$, and replace the pair
\[
    (x_{2i-1},x_{2i})
\]
by one of the other three possible two-bit strings, chosen uniformly at random.
All other bits are left unchanged. The resulting string has length \(k_i-1\), and therefore identifies a bin in
layer \(L_{i+1}\). This bin is \(h_2(x)\).

Equivalently, conditioned on
\[
    h_1(x)=x_1x_2\cdots x_{k_i},
\]
the hash \(h_2(x)\) is uniformly distributed over the three bins in layer
\(L_{i+1}\) whose addresses agree with $x_1x_2\cdots x_{k_i-1}$ in every position except possibly positions $2i-1$ and $2i$, and whose
two-bit string in those positions differs from the original pair $(x_{2i-1},x_{2i})$.

Of course, in order to ensure $h_2$ is well-defined, we need to make sure that the bit 
positions $2i-1$ and $2i$ \emph{exist} (i.e., that $2i \le k_i-1$). This is verified in the following lemma.

\begin{lemma}
The second hash function \(h_2\) is well-defined.
\end{lemma}

\begin{proof}
Suppose \(h_1(x)\) selects a bin in layer \(L_i\), where \(i<\ell\). The address
of this bin has length
$
    k_i = K+\ell-i.
$
After deleting the final bit, the remaining string has length \(k_i-1\). Thus
the bit positions \(2i-1\) and \(2i\) exist provided
\[
    2i\le k_i-1.
\]

The hardest case is \(i=\ell-1\). In that case,
$
    k_i-1
    =
    k_{\ell-1}-1
    =
    K.
$
So it suffices to show
\[
    2\ell-2\le K.
\]
By assumption,
$
    \ell=\frac14\log_2 n+O(1)
$
and
$
    K=\frac34\log_2 n-O(\log\log n).
$
Therefore, for all sufficiently large \(n\),
\[
    2\ell-2
    =
    \frac12\log_2 n+O(1)
    \le
    \frac34\log_2 n-O(\log\log n)
    =
    K.
\]
Hence \(2i\le k_i-1\) for every \(i<\ell\), and \(h_2\) is
well-defined.
\end{proof}

\paragraph{The full picture: braided cuckoo hashing.} Having described the bin layout and the two hash functions, we can now take a moment 
to consider what a braided cuckoo hash table actually looks like. Figure \ref{fig:four-layer-construction} shows a four-layer instance of the hash-table construction. Each dot represents a bin, and an edge from a bin $b_1$ in \(L_i\) to a bin $b_2$ in \(L_{i+1}\) indicates that the pair $(b_1, b_2)$ is a legal outcome for the hashes $(h_1(x), h_2(x))$ of an element $x$.

\begin{figure}[t]
    \centering
    \includegraphics[width=\textwidth]{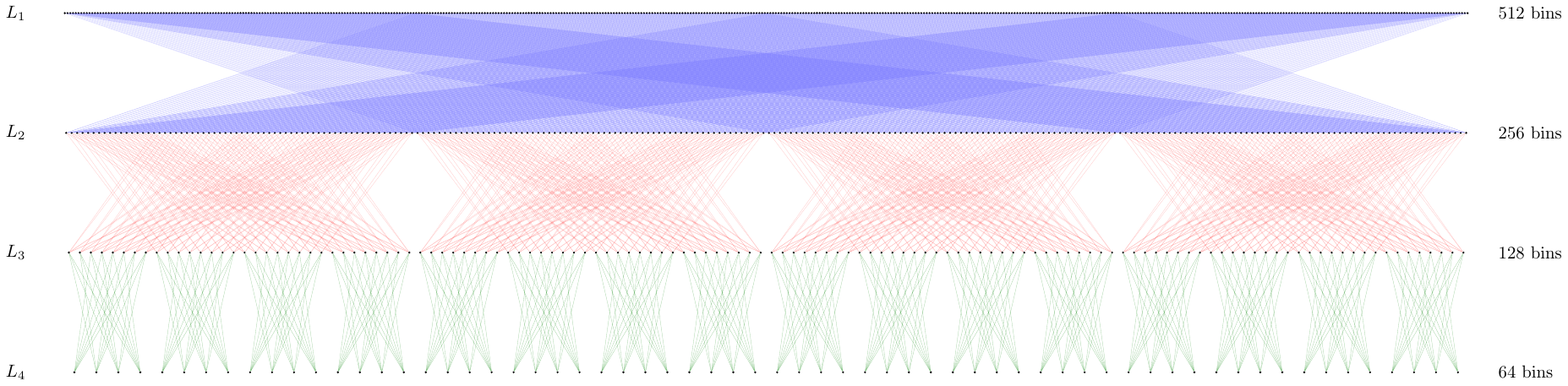}
    \caption{A four-layer instance of the hash-table construction with layers containing \(b_1 = 64\), \(b_2 = 32\), \(b_3 = 16\), and \(b_4 = 8\) bins. To represent a bin number as a binary string, we use the first bit as the most significant bit, and the last bit as the least significant bit.}
    \label{fig:four-layer-construction}
\end{figure}

Notice that the edges between layers interleave in a way that looks almost like braids -- this is where the name ``braided cuckoo hashing'' comes from.

It is worth remarking on how sparse the graph in Figure \ref{fig:four-layer-construction} is. Whereas, in traditional cuckoo-hashing constructions, each element chooses out of $\tilde{\Theta}(n^2)$ possible $(h_1(x), h_2(x))$ edges, braided cuckoo hashing has only $O(n / B)$ such pairs. In fact, each node in Figure \ref{fig:four-layer-construction} is incident to at most $9$ other nodes (6 from the previous layer, and 3 from the next layer). 

The fact that nodes have $O(1)$ degree will end up being critical for supporting insertions and deletions efficiently in Section \ref{sec:dynamic-placement}. At the same time, the sparseness of the graph makes it highly non-obvious whether it is \emph{even possible} to place $n$ elements together in the hash table. This will be the focus of the next section, and is where the specific ``braided structure'' of the table ends up being important.

\section{Static Placement}\label{sec:static-placement}

In this section we consider the static problem of placing $n$ elements into the braided cuckoo hash table. 
We will show that, despite the somewhat odd structure of the hash table, there is a simple way of placing
the elements into the hash table that succeeds with high probability. As in the previous section, we will 
make both the High-Load Assumption and the Divisibility Assumption throughout.

\paragraph{The static placement process. }
We now describe a simple placement process for a fixed set of \(n\)
elements. The process places elements layer by layer, using \(h_1\) as the
primary choice and \(h_2\) only when the primary bin is already full.

For an element \(x\), let
\[
    \lambda(x)\in\{1,\ldots,\ell-1\}
\]
denote the layer containing \(h_1(x)\). Recall that \(h_1\) never selects the
final layer \(L_\ell\). If \(\lambda(x)=i\), then \(h_2(x)\) is a bin in the
next layer \(L_{i+1}\).

Consider a static placement process that proceeds in phases \(i=1,2,\ldots,\ell-1\) as follows. During phase
\(i\), for each bin \(q\in L_i\), we consider the elements \(x\) with
\(h_1(x)=q\). The bin \(q\) may already contain elements forwarded from
\(L_{i-1}\); we place as many of the elements with \(h_1(x)=q\) into the
remaining slots of \(q\) as possible, and forward the rest to their
second-choice bins \(h_2(x)\in L_{i+1}\).

The choice of which elements are forwarded, out of those with $h_1(x)=q$, is made without looking at their $h_2$ hashes. So, conditioned on
being forwarded, each element's second hash is still independently and uniformly
distributed over the three allowed choices.

\paragraph{The Static Placement Theorem. } The main result of this section is the following theorem, which shows that the static placement process succeeds with high probability.

\begin{restatable}[Static placement theorem]{theorem}{staticplacement}
\label{thm:static-placement}
Suppose the constant $c$ used in the construction is sufficiently large. 
Then with high probability in $n$, the following events hold
simultaneously.

\begin{enumerate}
    \item For every layer \(L_i\) with \(2\le i\le \ell\), and every bin
    \(q\in L_i\), the number $F_q$ of elements forwarded into $q$ from layer $L_{i-1}$ satisfies
    \[
        \frac{B}{4}
        \le
        F_q
        \le
        \frac{3B}{4}.
    \]
    Notably, this means that no element is ever forwarded to a bin that is already full. 

    \item After all \(n\) elements have been placed, the only remaining free slots are in the final layer $L_\ell$.
\end{enumerate}
\end{restatable}

We will also need an extension of the Static Placement Theorem to the
underloaded regime, where the table (built for capacity parameter \(n\)) is
used to store \(n'\le n\) keys. For a bin \(q\in L_i\) with \(i\ge 2\), let
\(\mathsf{Pred}(q)\subseteq L_{i-1}\) denote the set of predecessor bins from which
elements may be forwarded to \(q\); the structure of \(h_2\) will force
\(|\mathsf{Pred}(q)|=6\). For \(q'\in\mathsf{Pred}(q)\), let \(H_{q',q}\) denote the number
of elements forwarded from \(q'\) to \(q\). The following theorem contains not just variations of the conclusions of the Static Placement Theorem, but also one more point that will be important: that every bin $q$ in layers $L_i$, $i > 1$, with probability at least $1-1/\poly(n)$, either has a free slot, or has $H_{q',q} = \Theta(B)$ elements forwarded to it from each of its predecessors \(q'\in\mathsf{Pred}(q)\). 

\begin{restatable}[Static placement with fewer than \(n\) keys]{theorem}{staticplacementunderloaded}
\label{thm:static-placement-underloaded}
Suppose the constant $c$ used in the construction is sufficiently large. Fix any \(n'\le n\), and run the static placement process on \(n'\) keys in the
table constructed for parameter \(n\). There is an absolute constant
\(\eta>0\) such that, with high probability in $n$, the following
events hold simultaneously.

\begin{enumerate}
    \item Every forwarded element is successfully placed in its second-choice
    bin, and for bins $q$ in layers $L_i$, $i > 1$, the number $F_q$ of elements forwarded into $q$ from layer $L_{i-1}$ satisfies $F_q \le \frac{3B}{4}$.

    \item Every bin in the final layer \(L_\ell\) has at least one free slot.

    \item For every layer \(L_i\) with \(2\le i\le \ell\), and every bin
    \(q\in L_i\), either \(q\) has a free slot after the placement process
    completes, or
    \[
        H_{q',q}\ge \eta B
        \qquad
        \text{for every } q'\in\mathsf{Pred}(q).
    \]
\end{enumerate}
\end{restatable}

\paragraph{Technical overview. } For brevity, we defer the proof of the Static Placement Theorem to Section \ref{sec:static-placement-proof}, and the proof of the static placement with fewer than $n$ keys theorem to Section \ref{sec:underloaded}. Here, instead, we give a high-level technical overview (focusing only on Theorem \ref{thm:static-placement}).

The main challenge is showing that no bin ever has more balls forwarded to it
than it can hold (or so few balls forwarded to it that it is unexpectedly underloaded). 
For each bin $q$, let its \emph{overflow} mean the number of balls
that $q$ forwards to the next layer (balls with $h_1(x) = q$ but that did not fit in $q$).
The main risk that we run is a potential positive-feedback loop, in which, as we look at bins
in later and later layers, the \emph{variance} of each bin's overflow increases. This would
eventually lead to some bins having an overflow that is either uncontrollably large 
(too large for the next-layer bins to handle) or too small (resulting in bins in the next layer being underloaded). 

Thus the core of the proof is to show that
the per-bin overflow variance is actually subject to a \emph{negative} feedback loop: variance
inherited from the previous layer is damped before it reaches the next layer,
and the new variance introduced from the current layer is only \(O(B)\). 
Once this variance control is established, it is not hard to strengthen the analysis to argue that, 
within each bin, the overflow is a sub-Gaussian random variable -- this implies that the overflow is 
tightly concentrated around its intended mean of $B/4$, which allows us to complete the proof. 

The proof proceeds by induction, where the inductive step considers (roughly) the following scenario. Suppose, heuristically, that every bin in layer
\(L_{i-1}\) overflows by roughly \(B/4\) elements, with variance at most \(CB\) for some large constant $C$; further suppose that these overflows are mutually independent (they are not, but we will come back to this later). Then, to perform the inductive step, we must consider some bin \(q\in L_i\) and argue that $q$'s overflow also has expectation roughly $B/4$ and variance at most $CB$. Since the calculation of the expectation is straightforward, we will focus here exclusively on the variance.

The definition of \(h_2\) is designed so that \(q\) has
exactly six possible predecessor bins in \(L_{i-1}\). From each predecessor bin,
an overflowed element chooses among three possible destinations in \(L_i\), one
of which is \(q\). Thus \(q\) receives, from each predecessor, roughly one third
of that predecessor's overflow.

Consider one predecessor bin, and
let \(X\) be its overflow, with mean roughly $B/4$ and variance at most $CB$. The number of these \(X\) overflowed elements that are sent to
the fixed bin \(q\) is obtained by independently keeping each of them with
probability \(1/3\). 

This sampling affects the variance in two ways. First, the fluctuation already present in \(X\) is scaled
down by a factor of \(3\), so its variance contribution is scaled down by
\(9\). Second, the random sampling itself introduces fresh variance.
Since \(X\) is typically \(O(B)\), this fresh variance is only \(O(B)\). Thus a
single predecessor contributes variance roughly
\[
    \frac{CB}{9}+O(B)
\]
to the load of \(q\).
Since \(q\) has six predecessors, the forwarded load into \(q\) has variance
roughly
\[
    6\left(\frac{CB}{9}+O(B)\right)
    =
    \frac23 CB +O(B).
\]

In addition to the elements forwarded from the previous layer, $q$ also receives elements that hash to it via $h_1$. The number of such elements has variance \(O(B)\). Therefore, the overflow of $q$, which is given by roughly
\[
    (\text{incoming forwarded load})
    +
    (\text{$h_1$ arrivals})
    -
    B
\]
has variance at most
\[
    \frac23CB+O(B).
\]
By choosing \(C\) sufficiently large, this is at most, say,
\[
    \frac56CB.
\]
Thus, critically, there is slack: going from one layer to the next, the variance is still bounded by \(CB\). This is the heart of the induction.

Finally, it is worth discussing why we can make the independence assumption used in this calculation (i.e., that overflows from the bins in the previous layer are mutually independent), as this is the main place where the specific ``braided structure'' of the table comes into play. Of
course, overflows from arbitrary bins in the same layer are not independent:
nearby bins may have interacted through earlier forwarding steps. The point of
the construction of \(h_2\) is that the six predecessor bins of a fixed bin
\(q\in L_i\) are not arbitrary nearby bins. They differ only in the bit-pair
that is used for the transition from \(L_{i-1}\) to \(L_i\), together with the
bit that is deleted in that transition. Transitions in earlier layers never toggled these bits. So, if we build for each predecessor bin $b$ a predecessor tree $T_b$ consisting of $b$'s predecessors in $L_{i - 2}$, their predecessors in $L_{i - 3}$, and so on, then we will find that the trees $T_b$ are disjoint across the six predecessor bins $b$. Consequently, for the purposes of analyzing the load forwarded
into \(q\), the six predecessor overflows may be treated as independent\footnote{Technically, they are still not independent due to the fact that the numbers of balls which hash to different parts of the table are not quite independent, but this point can easily be handled in the analysis using a Poissonization argument.}. The
obliviousness of the forwarding rule is also important here: conditioned on an
element being forwarded, its choice among the three possible destinations in the
next layer remains uniform and independent of the event that it was forwarded.

When we put the pieces together, we can prove within each layer that the overflow
 in each bin behaves roughly like a Gaussian random variable with mean $B/4$ and variance at most $CB$. 
 Because $B$ itself is a large constant multiple of $\log n$ (thanks to the High-Load Assumption), 
 the probability of a given bin behaving poorly (overflowing by too much or too little) is $1 / \poly(n)$.
 This is sufficiently small that we can union bound over all bins in order to prove the theorem. 

\section{Dynamic Placement}
\label{sec:dynamic-placement}

We now consider the dynamic setting: The goal is to support an arbitrary
oblivious sequence of
\[
    T=n^{O(1)}
\]
insertions and deletions, subject to the constraint that at most $n$ keys are present at any time.  
We want every operation to succeed with high probability in $n$ in time $O(\log n)$. Note that,
as in the previous sections, we continue to make both the High-Load Assumption and the Divisibility Assumption.

An interesting aspect of the algorithms in this section is the subtle role of
\emph{history independence} \cite{antipers}.  A data structure is \emph{strongly history
independent} if, at any given moment, the state of the data structure is
completely determined by the current set of keys, together with the fixed
randomness of any hash functions that the data structure uses.  It is
\emph{weakly history independent} if the probability distribution of the state
of the data structure is completely determined by the current set of keys,
again together with the fixed randomness of the hash functions.\footnote{For
example, a data structure that stores up to \(n\) keys in a prefix of an array of
size \(n\) would be strongly history independent if it stored the keys in some
canonical sorted order, and would be weakly history independent if it stored the
keys in a uniformly random order, supporting insertions and deletions by the
usual random-prefix insertion and deletion operations.}  History independence is
often studied as a security property \cite{antipers}: if the memory representation of the data
structure is leaked to an adversary, it does not reveal any unintended information
about the sequence of operations that produced the current set of keys.

In our setting, however, history independence plays a different role.  It gives
a way to keep the table in a continually well-controlled state, one that can still be
analyzed using the Static Placement Theorem.  More concretely, it prevents
\emph{reappearance dependencies}: problematic dependencies that can arise when an element is
inserted, deleted, and later reinserted, at which point its hash values have
already influenced the previous state of the table. In our context, we must be sensitive
to avoid any reappearance dependencies that might, for example, cause a bin $q'$ to evict elements
$x \mid h_1(x) = q'$ in a way that is biased towards evicting those with a specific second hash 
$h_2(x) = q$. This type of bias could cause the bin $q$ to become severely overloaded,
breaking the guarantees of the Static Placement Theorem. We will see that history independence
lets us dodge this issue, and more generally the issue of reappearance dependencies, entirely. 
To make this connection explicit, we will state in our theorems not only the running times of the
dynamic constructions below, but also their history-independence guarantees.

The rest of the section is split into two parts.  First, we present a simple warmup
construction that achieves $O(\log^2 n)$-time operations. The advantage of the warmup 
is that it is conceptually simple: a priority hash function canonically determines 
which keys overflow from each bin, guaranteeing that the resulting table is strongly 
history independent.  The drawback is that each step of an eviction chain or pullback 
chain costs \(\Theta(B)\) time, giving \(O(\ell B)=O(\log^2 n)\)-time operations. 
To eliminate this drawback, our second construction shows how to build a slightly more
subtle \emph{weakly} history-independent solution that, with the help of reservoir-sampling style techniques,
is able to support $O(\log n)$-time operations with high probability. 

We remark that, in addition to being pedagogically useful, the strongly history-independent warmup 
construction ends up being concretely useful later in the paper, allowing us to avoid potentially painful
 dependencies when rebuilding parts of the data structure in Section \ref{sec:deamortization}.

\paragraph{Conventions for discussing eviction chains.}
In both of our algorithms, insertions will perform what we call a \defn{forward eviction chain}. 
  A key first tries to enter its first-choice bin.
If that bin is full, the algorithm evicts one resident first-choice key from the
bin (possibly the key currently being inserted) and forwards the evicted key to its second-choice bin.  
This may cause another eviction in the next layer, and so on.

Deletions, on the other hand, will perform what we call a \defn{pullback chain}.
Removing a key leaves a hole in some bin \(q\).  If \(q\) has a key that had
overflowed out of \(q\), then one such key is pulled back into \(q\).  This
creates a hole in the pulled-back key's second-choice bin, and the process
continues.

Note that both types of chains are guaranteed to have length at most $\ell = O(\log n)$. 

\paragraph{Additional notation.} Finally, before we continue, it will also be helpful to introduce some 
notation that we will use throughout the rest of the section. Let $S$ be the current set of keys in the hash table.
For a bin \(q\), let
\[
    A_q(S)=\{x\in S:h_1(x)=q\}
\]
be the keys whose first-choice bin is \(q\).  Let \(I_q(S)\) be the keys
forwarded into \(q\) from the previous layer.  For \(q\notin L_\ell\), define
\[
    r_q(S)=\max\{0,\ |A_q(S)|+|I_q(S)|-B\}.
\]
The bin \(q\) must overflow exactly \(r_q(S)\) keys from \(A_q(S)\).  We write
\[
    V_q(S)\subseteq A_q(S),
    \qquad
    |V_q(S)|=r_q(S),
\]
for this overflowing set.  The bin \(q\) stores
\[
    I_q(S)\cup \bigl(A_q(S)\setminus V_q(S)\bigr),
\]
and each key in \(V_q(S)\) is stored in its second-choice bin.  Bins in the
final layer have no first-choice keys, since \(h_1\) never selects \(L_\ell\).

Finally, for a bin \(q\), we will let \(\mathsf{Succ}(q)\) denote the three possible
second-choice bins of keys whose first-choice bin is \(q\).

We remark that the sets \(A_q,I_q,V_q\) are used only in the analysis.  The data
structure stores no metadata.  Given a key in a slot, we can recompute its
\(h_1\)- and \(h_2\)-values to determine whether it is stored in its
first-choice or second-choice bin.

\subsection{An $O(\log^2 n)$-time Strongly History-Independent Warmup}

The first construction uses a priority hash function
\[
    \rho:U\to[0,1]
\]
to determine which keys get forwarded from a given bin $q$ to its successors in the next layer.
For a given bin $q$, we maintain the invariant that the $r_q(S)$ elements $V_q(S)$ that $q$ forwards 
to the next layer are precisely the $r_q(S)$ elements of \(A_q(S)\) with the smallest \(\rho\)-values.
This invariant completely determines which keys go to which bins. Then, within a bin, we 
store the keys in an arbitrary canonical order (say, lexicographically). The result is 
a strongly history-independent data structure: the state, at any given moment, is completely
determined by the current set $S$ of elements and the functions $h_1, h_2, \rho$.

\paragraph{Implementing insertions and deletions. }Having described the invariant that determines
the state of the data structure, we now discuss how to implement insertions and deletions 
in order to preserve the invariant.

The following helper procedures all touch only a constant number of bins, and
therefore cost \(O(B)\) time.

\[
\begin{array}{ll}
\textsc{Find}(x)
    & \text{return }q \in \{h_1(x), h_2(x)\} \text{ such that }x \text{ is stored in }q;\\[1mm]

\textsc{CanonicalInsert}(q,x)
    & \text{insert }x\text{ into }q\text{ and restore the canonical order};\\[1mm]

\textsc{CanonicalDelete}(q,x)
    & \text{delete }x\text{ from }q\text{ and restore the canonical order};\\[1mm]

\textsc{ResidentPrimaries}(q)
    & \text{return }\{x\text{ stored in }q:h_1(x)=q\};\\[1mm]

\textsc{OverflowedFrom}(q)
    & \text{return }\{x\text{ stored in a bin of }\mathsf{Succ}(q):h_1(x)=q\}.
\end{array}
\]

Using these helper functions as subroutines, insertions and deletions can be implemented as follows. 
 The two victim-selection procedures encode
the priority rule: an insertion evicts the lowest-priority eligible
first-choice key, while a deletion pulls back the highest-priority key that had
previously overflowed.  The insertion and deletion algorithms then simply apply
these procedures along the forward eviction chain or the pullback chain, respectively.

\begin{algorithm}[H]
\caption{\(\textsc{SHI-EvictionVictim}(q,y)\)}
\begin{algorithmic}[1]
\Require A full bin \(q\), and an incoming key \(y\) trying to enter \(q\).
\If{\(h_1(y)=q\)}
    \State \(C\gets \textsc{ResidentPrimaries}(q)\cup\{y\}\).
\Else
    \State \(C\gets \textsc{ResidentPrimaries}(q)\).
\EndIf
\If{\(C=\emptyset\)}
    \State The table has failed.
\EndIf
\State \Return the key \(z\in C\) minimizing \(\rho(z)\).
\end{algorithmic}
\end{algorithm}

\begin{algorithm}[H]
\caption{\(\textsc{SHI-PullbackVictim}(q)\)}
\begin{algorithmic}[1]
\State \(C\gets \textsc{OverflowedFrom}(q)\).
\If{\(C=\emptyset\)}
    \State \Return \(\bot\).
\Else
    \State \Return the key \(z\in C\) maximizing \(\rho(z)\).
\EndIf
\end{algorithmic}
\end{algorithm}

\begin{algorithm}[H]
\caption{\(\textsc{SHI-Insert}(x)\)}
\begin{algorithmic}[1]
\State \(y\gets x\).
\State \(q\gets h_1(y)\).
\While{true}
    \If{\(q\) has a free slot}
        \State \(\textsc{CanonicalInsert}(q,y)\).
        \State \Return
    \EndIf
    \State \(z\gets \textsc{SHI-EvictionVictim}(q,y)\).
    \If{\(z\neq y\)}
        \State \(\textsc{CanonicalDelete}(q,z)\).
        \State \(\textsc{CanonicalInsert}(q,y)\).
    \EndIf
    \State \(y\gets z\).
    \State \(q\gets h_2(y)\).
\EndWhile
\end{algorithmic}
\end{algorithm}

\begin{algorithm}[H]
\caption{\(\textsc{SHI-Delete}(x)\)}
\begin{algorithmic}[1]
\State \(q \gets \textsc{Find}(x)\).
\State \(\textsc{CanonicalDelete}(q,x)\).
\While{\(q\notin L_\ell\)}
    \State \(z\gets \textsc{SHI-PullbackVictim}(q)\).
    \If{\(z=\bot\)}
        \State \Return
    \EndIf
    \State \(q^+\gets h_2(z)\).
    \State \(\textsc{CanonicalDelete}(q^+,z)\).
    \State \(\textsc{CanonicalInsert}(q,z)\).
    \State \(q\gets q^+\).
\EndWhile
\end{algorithmic}
\end{algorithm}

\paragraph{Analysis. } Having described the insertion and deletion algorithms, we
can now analyze the algorithm. Critically, thanks to history independence,
the state of the data structure at any given moment can be reasoned about through the 
 Static Placement Theorem. 

\begin{theorem}[Strongly History-Independent Dynamic Placement]
\label{thm:dynamic-shi}
Suppose that the parameters $n$, $B$, $m$ satisfy the Divisibility Assumption from
Subsection \ref{sec:divisibility-assumption}. Consider an oblivious sequence of $T = n^{O(1)}$ 
insertions and deletions that never has more than $n$ live keys at a 
 time; and assume that the constant $c$ used by the algorithm is sufficiently large as a function of
  $\log_n T$.

Then the algorithm above uses $m = n + O(n^{0.75})$ slots, is strongly history-independent,
and supports the full insertion/deletion sequence without failure with high probability in $n$. 
Each insertion or deletion takes time $O(\log^2 n)$, while queries take time $O(\log n)$. 
\end{theorem}
\begin{proof}
It is straightforward to verify that the insertion and deletion procedures above 
correctly maintain the invariant that, if the current set of elements is $S$, then
the overflowing set $V_q(S)$ of each bin $q$ consists of the elements of $A_q(S)$ with the
 smallest $\rho$-values; and that the only bins to have free slots are the bins $q$
 with no overflow elements, meaning $V_q(S) = \emptyset$. It follows that the state
 of the hash table, at any given moment, is precisely what one obtains by applying
 the static placement process from Section \ref{sec:static-placement}, with the rule
 that, when selecting elements from $A_q(S)$ to overflow out of bin $q$, we select
 those with the smallest $\rho$-values. 

This implies two things: (1) that the data structure is strongly history independent, as claimed,
and (2) that we can apply the Static Placement Theorem with $n' < n$ keys 
(Theorem \ref{thm:static-placement-underloaded}) to reason about the state of the hash table 
at any given moment. From Theorem \ref{thm:static-placement-underloaded}, we know that, with high probability in $n$, every bin $q$ has at most $\frac{3B}{4}$ elements forwarded to it 
from each of its predecessors, so any bin $q$ which is full has at least one element $x$ with $h_1(x) = q$. 
Taking a union bound over the $T = n^{O(1)}$ updates, we know that, with high probability in $n$, 
this is true before every insertion. Thus, critically, the failure condition in \textsc{SHI-EvictionVictim} 
can never be triggered
-- every bin that we wish to find an insertion-eviction victim in will have at least one valid victim, namely
some element $x$ with $h_1(x) = q$. 
Since this is the only way in which the hash table can 
incur a failure, the probability of failure is at most $1/\poly(n)$. 

Finally, should no failure occur, the worst-case time bounds are determined as follows. 
Each insertion or deletion interacts with at most $O(\ell) = O(\log n)$ bins, 
and spends time $O(B) = O(\log n)$ interacting with each one. Therefore, the worst-case
time per insertion or deletion is $O(\ell B) = O(\log^2 n)$. Queries, on the other hand,
take worst-case time $O(B) = O(\log n)$, since they need only check two bins.

\end{proof}

\subsection{An $O(\log n)$-time Solution via Weak History Independence} 

The reason that the algorithm in the previous section is unable to achieve
$O(\log n)$-time operations is that, each time it is selecting a key to evict or pull back,
it must find the \emph{specific key} that the strongly history-independent algorithm
requires it to move. This forces the algorithm to either spend $\Omega(B)$ time
scanning the bin to find the key (if the bin is unsorted) or $\Omega(B)$ time 
keeping the bin in sorted order each time the bin is modified (if the bin is sorted, as in the
implementation given in the previous section). Either way, the algorithm 
spends $\Omega(B)$ time at every step of the eviction or pullback chain.

In this section, we will see that by using a \emph{weakly} history-independent data structure,
as opposed to a \emph{strongly} history-independent one, we can reduce the time per operation
to $O(\log n)$. The basic idea is that, instead of finding a specific key to move (during an eviction or pullback),
it suffices to instead make sure that the key which we move is chosen from a certain probability distribution;
and, if we do this correctly, we can still maintain (weak) history independence, but with $O(1)$ time (on average) per 
eviction or pullback step. 

Formally, the weakly history-independent data structure maintains the following distributional invariant:

\begin{quote}
\textbf{The Uniform Overflow Invariant:} For every bin \(q\), conditioned on the placements in earlier layers and on the
value \(r_q\), the overflowing set \(V_q\) is a uniformly random
\(r_q\)-subset of \(A_q\) (and this is independent across bins $q$ in the same layer).
\end{quote}

Note that, if we have this invariant, then we can still apply the Static Placement Theorem,
because the state of the data structure at any given moment is precisely what one obtains by applying
the static placement process from Section \ref{sec:static-placement}, with the rule
that, when selecting elements from $A_q(S)$ to overflow out of bin $q$, we select 
a uniformly random subset. 

In addition to maintaining the Uniform Overflow Invariant, we will also maintain
a second invariant that, within each bin, the keys are kept in a uniformly random order (independently across bins);
and if the bin has free slots, then those slots all appear at the end of the bin.\footnote{We remark that this second invariant is \emph{not necessary} for the correctness or time bounds of the algorithm,
but it is technically necessary in order for us to declare that the resulting 
data structure is weakly history-independent.
It is also slightly convenient algorithmically, as it lets us test in \(O(1)\) time whether a bin is full, by
checking only its final slot.}

\paragraph{The insertion and deletion algorithms. }At a high level, the algorithm works as follows.  Whenever an insertion needs
to evict a key forward, or a deletion needs to pull a key backward, we use a
sampling procedure to choose the key that moves.  The sampling procedures are
designed so that the chosen key has exactly the distribution needed to preserve
the Uniform Overflow Invariant. 

The procedures for finding a key to evict or pull back will now use \emph{rejection sampling} 
to find the desired key. Using the properties from the (underloaded) Static Placement Theorem
(Theorem \ref{thm:static-placement-underloaded}), we will be able to argue that each
sampling trial has probability $\Omega(1)$ of success. This gives an expected-time bound of 
$O(1)$ per eviction/pullback, and applying a Chernoff bound over all the eviction/pullback steps
we will be able to recover a high-probability time bound of $O(\log n)$.

We now give the procedures used to implement insertions and deletions.  First, the following two routines maintain the
random occupied-prefix representation of each bin.

\begin{algorithm}[H]
\caption{\(\textsc{RandomPrefixInsert}(q,x,k)\)}
\begin{algorithmic}[1]
\Require Bin \(q\) currently contains \(k<B\) keys in a random occupied prefix.
\State Choose \(j\in\{1,\ldots,k+1\}\) uniformly at random.
\If{\(j\le k\)}
    \State Move the key in slot \(j\) to slot \(k+1\).
\EndIf
\State Store \(x\) in slot \(j\).
\end{algorithmic}
\end{algorithm}

\begin{algorithm}[H]
\caption{\(\textsc{RandomPrefixDelete}(q,s,k)\)}
\begin{algorithmic}[1]
\Require Bin \(q\) currently contains \(k\) keys in a random occupied prefix,
and slot \(s\le k\) is to be deleted.
\If{\(s<k\)}
    \State Move the key in slot \(k\) to slot \(s\).
\EndIf
\State Clear slot \(k\).
\end{algorithmic}
\end{algorithm}

We also use the following basic helpers.

\[
\begin{array}{ll}
\textsc{Full}(q)
    & \text{return whether bin }q\text{ is full, by checking whether slot }B\text{ is occupied};\\[1mm]

\textsc{Load}(q)
    & \text{return the number of elements in bin }q\text{, through a linear scan of the slots};\\[1mm]

\textsc{Find}(x)
    & \text{return the pair }(q, s)\text{ such that }x\text{ is in slot }s\text{ of bin }q.
\end{array}
\]

Next come the sampling procedures used to find keys for evictions and pullbacks. 
 The first procedure samples a uniformly random
resident first-choice key of a full bin.  This is the operation needed at an
internal step of a forward eviction chain.  The second handles the first step
of an insertion, where the new key itself must also be included in the sample.
The third samples a uniformly random key that has overflowed out of a bin \(q\);
this is the operation needed in a pullback chain.

\begin{algorithm}[H]
\caption{\(\textsc{SampleResidentPrimary}(q)\)}
\begin{algorithmic}[1]
\Require A full bin \(q\).
\For{\(t=1,\ldots,B\)}
    \State Choose a uniformly random slot \(s\in\{1,\ldots,B\}\) of \(q\).
    \If{slot \(s\) contains a key \(z\) with \(h_1(z)=q\)}
        \State \Return \((z,s)\).
    \EndIf
\EndFor
\State Scan slots \(s=1,\ldots,B\) of \(q\) for a key \(z\) with \(h_1(z)=q\).
\If{no such key exists}
    \State The table has failed.
\EndIf
\State \Return \((z,s)\).
\end{algorithmic}
\end{algorithm}

\begin{algorithm}[H]
\caption{\(\textsc{SamplePrimaryOrNew}(q,y)\)}
\begin{algorithmic}[1]
\Require A full bin \(q\), and a new key \(y\) with \(h_1(y)=q\).
\Repeat
    \State Choose uniformly from the following \(4B+1\) locations:
    \Statex \hspace{1cm} the \(B\) slots of \(q\), the \(3B\) slots in the bins of \(\mathsf{Succ}(q)\), and the singleton location \(\{y\}\).
\Until{the chosen location is \(\{y\}\), or contains a key \(z\) with \(h_1(z)=q\)}
\State \Return the selected key and its location. \Comment{If \(\{y\}\) was chosen, this is \((y,\{y\})\).}
\end{algorithmic}
\end{algorithm}

\begin{algorithm}[H]
\caption{\(\textsc{SampleOverflow}(q)\)}
\begin{algorithmic}[1]
\Require A bin \(q\) that has just lost a key.
\If{every bin in \(\mathsf{Succ}(q)\) is nonfull}
    \State Scan the bins in \(\mathsf{Succ}(q)\), collecting all keys \(z\) with \(h_1(z)=q\).
    \If{no such key exists}
        \State \Return \(\bot\).
    \Else
        \State \Return a uniformly random collected key, together with its slot.
    \EndIf
\Else
    \Repeat
        \State Choose a uniformly random slot \(s\) among the \(3B\) slots in the bins of \(\mathsf{Succ}(q)\).
    \Until{slot \(s\) contains a key \(z\) with \(h_1(z)=q\)}
    \State \Return \((z,s)\).
\EndIf
\end{algorithmic}
\end{algorithm}

The insertion algorithm follows a forward eviction chain.  At the first step,
the incoming key is a new first-choice key, so we sample from
\(A_q\cup\{x\}\).  At all later steps, the incoming key was forwarded from the
previous layer, so we sample only from the resident first-choice keys of the
current bin.

\begin{algorithm}[H]
\caption{\(\textsc{WHI-Insert}(x)\)}
\begin{algorithmic}[1]
\State \(y\gets x\).
\State \(q\gets h_1(y)\).
\While{true}
    \If{\(\neg\textsc{Full}(q)\)}
        \State \(k\gets \textsc{Load}(q)\).
        \State \(\textsc{RandomPrefixInsert}(q,y,k)\).
        \State \Return
    \EndIf

    \If{\(h_1(y)=q\)}
        \State \((z,\mathrm{loc})\gets \textsc{SamplePrimaryOrNew}(q,y)\).
        \If{\(\mathrm{loc}\) is a slot of \(q\)}
            \State Replace \(z\) by \(y\) in that slot.
            \State \(y\gets z\).
        \Else
            \State \(y\gets y\). \Comment{\(y\) itself is the newly overflowed key.}
        \EndIf
    \Else
        \State \((z,s)\gets \textsc{SampleResidentPrimary}(q)\).
        \State Replace \(z\) by \(y\) in slot \(s\).
        \State \(y\gets z\).
    \EndIf

    \State \(q\gets h_2(y)\).
\EndWhile
\end{algorithmic}
\end{algorithm}

Deletion creates a pullback chain.  If the bin from which the key is deleted was
already nonfull, then it had no overflowed keys of its own and the operation
stops after the deletion.  Otherwise, the bin has just acquired a hole, and we
try to fill that hole by pulling back a uniformly random key that had overflowed
out of it.

\begin{algorithm}[H]
\caption{\(\textsc{WHI-Delete}(x)\)}
\begin{algorithmic}[1]
\State \((q,s)\gets \textsc{Find}(x)\).
\If{\(\textsc{Full}(q)\)}
    \State \(\textsc{RandomPrefixDelete}(q,s,B)\).
    \While{\(q\notin L_\ell\)}
        \State \(a\gets \textsc{SampleOverflow}(q)\).
        \If{\(a=\bot\)}
            \State \Return
        \EndIf
        \State Let $(z, s)$ be the entries of $a$, and let \(q^+=h_2(z)\).
        \State \(\textsc{RandomPrefixInsert}(q,z,B - 1)\).
        \If{\(\textsc{Full}(q^+)\)}
            \State \(\textsc{RandomPrefixDelete}(q^+,s,B)\).
            \State \(q\gets q^+\).
        \Else
            \State \(\textsc{RandomPrefixDelete}(q^+,s,\textsc{Load}(q^+))\).
            \State \Return
        \EndIf
    \EndWhile
\Else
    \State \(\textsc{RandomPrefixDelete}(q,s,\textsc{Load}(q))\).
    \State \Return
\EndIf
\end{algorithmic}
\end{algorithm}

\paragraph{Analysis. } Having described the insertion and deletion algorithms, we
can now analyze the algorithm. We first argue that the sampling procedures correctly
preserve the Uniform Overflow Invariant.

\begin{lemma}
At any given moment, if the hash table has not failed, then the Uniform Overflow Invariant holds: for every bin $q$, conditional on the placements in earlier layers, the overflowing set $V_q$ is a uniformly random $r_q$-subset of the set $A_q$ consisting of elements with first-hash $q$ (and this is independent across bins $q$ in the same layer). \label{lem:uniform-overflow-invariant}
\end{lemma}
\begin{proof}
    It suffices to prove that insertions and deletions preserve the invariant.
    Moreover, since the algorithm only lets a bin $q$ have free slots if the bin
    has no overflowed keys (no keys with $h_1(x) = q$ are in the next layer),
    the number of overflowed keys from each bin $q$ will always be 
    $|V_q| = \max\{0, |A_q| + |I_q| - B\} = r_q$. So the only thing that we need
    to check is that the $r_q$ keys in $V_q$ are a uniformly random subset of those in $A_q$. 
    Throughout the rest of the proof, to avoid ambiguity, we will use $A_q$, $V_q$, and $r_q$ to refer
    to the values that these quantities had \emph{before} the current operation.

    Consider the first step of an insertion, where a new key \(x\) enters its
    first-choice bin \(q = h_1(x)\). Suppose \(q\) is full. After inserting \(x\), the overflowing set should
    be a uniformly random \((r_q + 1)\)-subset of \(A_q \cup \{x\}\). The procedure
    \(\textsc{SamplePrimaryOrNew}\) samples a uniform key
    \[
        z\in A_q\cup\{x\}.
    \]
    If \(z \in V_q \cup \{x\}\), the new key \(x\) becomes the newly overflowed key. If
    \(z \in A_q \setminus V_q\), then \(z\) becomes the newly overflowed key. This correctly places \(x\) into the overflow set with probability \((|V_q| + 1)/(|A_q| + 1)\), and the rest of the time it moves a uniformly random element in \(A_q \setminus V_q\) to \(V_q\); so, overall, the invariant that \(V_q\) is uniformly random from \(A_q\) is preserved.

    At every later step of the eviction chain, the incoming key is forwarded from
    the previous layer.  Thus \(A_q\) itself has not changed, but \(r_q\) increases
    by one.  Starting from a uniform \(r_q\)-subset \(V_q\subseteq A_q\), the
    correct update is to add a uniformly random key from \(A_q\setminus V_q\).
    These are exactly the resident first-choice keys of \(q\), and
    \(\textsc{SampleResidentPrimary}(q)\) samples one uniformly. Therefore, each 
    subsequent step of the eviction chain preserves the invariant.
    
    Now consider deletion. If a deletion removes a key $x$ with $h_1(x) = q$, 
    then the removal of $x$ (before we begin any pullback chain) preserves the
    invariant that the overflowed elements from $q$ are a uniformly random subset
    of the elements $y$ present with $h_1(y) = q$ (now $A_q \setminus \{x\}$).
    
    Next, we verify that the pullback chain preserves the invariant. For pullback to
    occur, the bin $q$ from which $x$ is removed must have overflowed keys, meaning 
    that \(r_q\) decreases by one, and the correct update is to remove a uniformly
    random key from \(V_q\), and to bring that element back to the bin $q$.  
    The procedure \(\textsc{SampleOverflow}(q)\) does
    exactly this, so the invariant is once again preserved.
\end{proof}

We also observe that, within each bin, the placement of keys is uniformly random
within some prefix of the bin. As discussed earlier, this invariant is not necessary
for the correctness or running time guarantees, but it is necessary to declare the algorithm weakly history-independent.
\begin{lemma}
At any given moment, if a bin $q$ contains $j$ elements, then they are placed in a uniformly random order in the
first $j$ slots of a bin (and this is independent across bins).
\end{lemma}
\begin{proof}
By design, any call to \(\textsc{RandomPrefixInsert}\) or \(\textsc{RandomPrefixDelete}\) 
preserves the invariant. More subtly, when either an eviction chain or a pullback chain
replaces a sampled victim \(z\) by an incoming key \(y\) in a full bin, the choice of 
who $z$ is is independent of the order of keys within the bin. Conditioned on the identity of $z$,
the slot of \(z\) is uniformly random. Given a uniformly random permutation of $j$
elements $Q$, if we overwrite a fixed one of those elements $z$ with a new element $y$, 
the resulting permutation is still uniformly random over the elements in $Q \setminus \{z\} \cup \{y\}$.
Thus, again, the invariant is preserved.
\end{proof}

\begin{theorem}[Weakly History-Independent Dynamic Placement]
\label{thm:dynamic-whi}
Suppose that the parameters $n$, $B$, $m$ satisfy the Divisibility Assumption from
Subsection \ref{sec:divisibility-assumption}. Consider an oblivious sequence of $T = n^{O(1)}$ 
insertions and deletions that never has more than $n$ live keys at a 
time; and assume that the constant $c$ used by the algorithm is sufficiently large as a function of
$\log_n T$.

Then the algorithm above uses $m = n + O(n^{0.75})$ slots, is weakly history-independent,
and supports the full insertion/deletion sequence without failure with high probability in $n$. 
Each insertion or deletion takes time $O(\log n)$ with high probability in $n$, while queries take time $O(\log n)$
in the worst case. 
\end{theorem}

\begin{proof}
    The preceding lemmas tell us that the algorithm is weakly history-independent. And, more
    importantly, that the state of the hash table at any given moment is drawn from the same
    distribution as if we were to rebuild the hash table from scratch with the same set of keys,
    using the static placement process, and with the rule that each overflowing bin forwards a uniformly
    random subset of its keys to the next layer. Thus, we can apply the (underloaded) Static Placement Theorem
    (Theorem \ref{thm:static-placement-underloaded}) to conclude that, at any given moment, 
    with high probability in $n$, the following events hold:
    \begin{enumerate}
    \item Every bin $q$ has at most $3B/4$ keys forwarded into it from the previous layer.
    \item Every bin $q$ in a layer $L_i$, $i \ge 2$, either has a free slot or
    satisfies $H_{q', q} \ge \eta B$ for each of its predecessors $q' \in \mathsf{Pred}(q)$.
    (Recall that $H_{q', q}$ is the number of keys forwarded from $q'$ to $q$.)
    \end{enumerate}
    Since we are interested only in a polynomial number $T$ of insertions/deletions, we have 
    with high probability in $n$ that these hold before and after every operation. 
    
    To complete the proof, we argue that, if the above events hold before a given insertion or deletion,
    then the operation succeeds (with probability 1) and takes time $O(\log n)$ with high 
    probability in $n$. 

    First, note that the only way for the operation to fail is if the failure condition
    is triggered in \textsc{SampleResidentPrimary}. This can only happen if some bin $q$
    becomes completely occupied by elements forwarded from the previous layer. But, this
    means that, before the operation, we had $F_q \ge B - 1 > 3B/4$, which is 
    a contradiction.

    Finally, we argue that the operation takes time $O(\log n)$ with high 
    probability in $n$. 
    To do this, we must reason about the time taken by each of the sampling procedures.
    Specifically, for a given sampling procedure, what is the probability that each of its
    sampling trials succeeds (conditioned on everything else that has happened so far)?

    First, consider \(\textsc{SampleResidentPrimary}(q)\).  This procedure is used
only when \(q\) is full (and has not been edited yet by the current operation).  Since at most \(3B/4\) of the keys in \(q\) were
forwarded from the previous layer, at least \(B/4\) slots contain resident
first-choice keys of \(q\).  Therefore each trial succeeds with probability at
least \(1/4\).

The same observation handles \(\textsc{SamplePrimaryOrNew}(q,y)\).  The sample
space has size \(4B+1\): the bin \(q\), the three successor bins of \(q\), and
the singleton location \(\{y\}\).  Since \(q\) is full and contains at least
\(B/4\) resident first-choice keys of \(q\), each trial succeeds with
probability \(\Omega(1)\).

The only subtle sampler is \(\textsc{SampleOverflow}(q)\).  Suppose first that
every successor of \(q\) is nonfull.  Then the procedure scans the three
successor bins.  This costs \(O(B)\), but it is necessarily the final step of
the pullback chain.  Indeed, after a key is pulled out of a nonfull successor
\(q^+\), that successor still has a free slot, so the pullback chain stops. 

Now suppose that at least one successor \(q^+\in\mathsf{Succ}(q)\) is full.
Since \(q\) is a predecessor of \(q^+\), 
\[
    H_{q,q^+}\ge \eta B.
\]
Thus at least \(\eta B\) of the \(3B\) slots in the successor bins contain keys
overflowed from \(q\).  Each rejection-sampling trial in
\(\textsc{SampleOverflow}(q)\) therefore succeeds with probability at least
\(\eta/3=\Omega(1)\).

Consequently, apart from the final step in a pullback chain, every
trial made by every sampler has conditional success
probability \(\Omega(1)\).  An eviction or pullback chain has length at most
\(\ell\), so an operation needs at most \(O(\ell) = O(\log n)\) successful samples.  A
Chernoff bound tells us that, for any $t = \Theta(\log n)$, if we flip coins that
each independently have $\Omega(1)$ probability of heads, then the number of flips
needed to get $t$ heads is at most $O(t + \log n) = O(\log n)$ with high probability in $n$. 
Thus the total time spent in sampling procedures (outside of the final 
call to \(\textsc{SampleOverflow}\)) is at most $O(\log n)$.

The remaining costs are \(O(B) = O(\log n)\) per operation, as they are all incurred
in either the first or final step of the chain, and from procedures that take
time at most $O(B)$ in the worst case. The theorem therefore follows.
\end{proof}

\section{Removing the Divisibility Assumption}
\label{sec:general-capacities}

So far, all of our results have been stated under the Divisibility Assumption 
from Subsection \ref{sec:divisibility-assumption}. This assumption states that
the number $n$ can be expressed as 
$$n = B2^K(2^\ell - 3/2),$$
where $\ell$ and $K$ are integers satisfying 
$\ell = \frac{1}{4}\log_2 n \pm O(1)$ and $K = \frac{3}{4}\log_2 n - \log B \pm O(1)$.
The assumption allowed us to design each layer of the hash table to have an exactly power-of-two number of bins.

In this section, we argue that the Divisibility Assumption is 
essentially WLOG, so long as we are willing to consider load factors $1 - \epsilon$ of at most, say, $1 - \Omega(n^{-1/5})$. 
The basic idea is that, given any capacity $n$,
we can break the
table into equal- and polynomial-size sub-tables (``blocks'') for which the arithmetic assumptions
hold, and we can hash each key to a block. We can then apply a Chernoff
to show that each block, with high probability, receives very close to the 
number of keys it is supposed to. And, finally, we can treat each of
the blocks as a separate hash table satisfying the Divisibility Assumption.

The following theorem, whose proof is deferred to Appendix \ref{app:general-capacities},
makes this reduction explicit, and removes the Divisibility Assumption from the results
in the previous section.

\begin{restatable}[Arbitrary capacities]{theorem}{generalcaps}
\label{thm:general-capacities}
There is a constant \(\alpha>0\), in fact one may take \(\alpha=1/5\), such that:
 For every sufficiently large \(n\), there exists a hash table construction
 with the same guarantees as Theorem \ref{thm:dynamic-whi} (or, if one wishes, Theorem \ref{thm:dynamic-shi}),
 but using $m = n + O(n^{1-\alpha})$ slots. 
\end{restatable}

\section{Supporting Arbitrary Load Factors}
\label{sec:all-load-factors}

So far, the data structures in the paper have been tuned to operate at load
factor \(1-1/n^{\Omega(1)}\), while supporting insertions, deletions, and
queries in time \(O(\log n)\).  In this section we explain how to support load
factor \(1-\epsilon\), for an arbitrary parameter \(\epsilon\ge n^{-\alpha}\) (where
$\alpha$ is some absolute constant determined by the theorem below),
with time-per-operation $O(\log \epsilon^{-1})$ rather than $O(\log n)$.

\begin{theorem}[Load factor \(1-\epsilon\)]
\label{thm:epsilon-load-factor}
There are absolute constants \(\epsilon_0>0\) and \(\alpha>0\) such that, for
every \(n^{-\alpha}\le \epsilon\le \epsilon_0\), there is an open-addressed hash
table using \(n+O(\epsilon n)\) slots, and that can support any oblivious insertion/deletion sequence of length
\(n^{O(1)}\), with up to $n$ live keys at a time, with the following guarantees.    With probability at least \(1-1/\poly(n)\),
every operation succeeds.  Queries take worst-case time
\(O(\log\epsilon^{-1})\), and insertions and deletions take expected amortized
time \(O(\log\epsilon^{-1})\). 

The only source of amortization is a global
rebuild performed once every \(n\) updates, which is also permitted to use up to
$n$ additional slots of temporary space. In addition to the slots themselves, 
the hash table stores $O(\log n)$
machine words of metadata. 
\end{theorem}

The amortization (and the large amount of temporary space used during rebuilds) will both
be removed in Section \ref{sec:deamortization}. 

The constants used in the proof are not optimized.  Throughout the section we
will assume that $\alpha$ and $\epsilon_0$ are sufficiently small positive constants.

\paragraph{Why a direct retuning is not enough.}
To prove the theorem, a tempting approach would be to
directly modify the construction from the previous section to use
bins of size $B = \Theta(\log\epsilon^{-1})$ instead of
$B=\Theta(\log n)$. This would give the desired time and space bound \emph{if the
construction succeeded}, but the resulting table would fail with substantial
probability, even for static placement.

Indeed, consider a compatible pair of bins \((q_1,q_2)\) that can arise as the
first and second choices of a key.  The probability that a fixed key has this
particular pair of choices is \(\Theta(B/n)\), so the number of keys with this
pair of choices has mean \(\Theta(B)\).  If
\(B=\Theta(\log\epsilon^{-1})\), then the probability that this pair receives
more keys than the \(2B\) slots available in \(q_1\cup q_2\) is 
\(\epsilon^{\Theta(1)}\) (rather than \(n^{-\Theta(1)}\)). Since there are
\(\Theta(n/B)\) possible pairs $(q_1,q_2)$, we expect roughly
\((n/B)\poly(\epsilon)\) such pairs to be overloaded, which is much larger than
\(1\) for natural values of \(\epsilon\).

Thus, we cannot hope to prove Theorem \ref{thm:epsilon-load-factor} by 
simply reusing the construction from the previous section. 

\paragraph{This section: Using a frontyard/backyard structure to prove Theorem \ref{thm:epsilon-load-factor}. }
Instead, we take a slightly more elaborate approach which, at a high level,
is as follows.

We break the data structure into a ``frontyard'' and a ``backyard'' \cite{arbitman2010backyard,bender2023iceberg}. 
The frontyard, which consists of the vast majority of the slots in the table,
 is divided into independent mini-structures of size
\(\poly(\epsilon^{-1})\), each of which is implemented using the construction 
from Theorem \ref{thm:general-capacities} (applied to Theorem \ref{thm:dynamic-whi}). Each key hashes to a random mini-structure,
and by default insertions/deletions are performed on that mini-structure.
If, however, an operation takes too long (time at least $L$, where $L$
is a sufficiently large constant multiplied by $\log\epsilon^{-1}$) or incurs 
a failure, then the operation 
is aborted in the mini-structure and completes in the backyard. The backyard can therefore be viewed as a ``fallback'' data structure
to handle failures (either time failures or genuine failures) that occur in the frontyard.

Additionally, once every \(n\) updates, we rebuild the entire
structure from scratch, allowing us to clear out unnecessary elements from the 
backyard, and to ``repair'' mini-structures that have incurred failures (but 
that are ready to be used again). 

We will describe the role of the frontyard and the backyard in more detail below.
The key insight in the analysis will be to show that (with high probability in $n$)
the number of slots that we need for the backyard is very small ($n \poly(\epsilon)$),
so the overall space usage remains at most $n + O(\epsilon n)$ slots.

\subsection{Data Structure Description}

We now present the full description of the data structure.

\paragraph{The frontyard. }
Let \(N=\epsilon^{-5}\), where for convenience we assume $N$ is an integer
(rounding does not affect the argument that follows at all). 

Each frontyard mini-structure is an independent copy of the 
data structure from Theorem \ref{thm:general-capacities} (applied to Theorem \ref{thm:dynamic-whi}),
 with capacity parameter \(N\).  Thus one mini-structure uses \(M=N+O(N^{4/5})\) slots,
 supports $O(\log\epsilon^{-1})$-time queries, and supports any $\poly(\epsilon^{-1})$-length 
 sequence of insertions and deletions with at most $N$ live keys in time $O(\log\epsilon^{-1})$
 per operation, with high probability in $\epsilon^{-1}$. 

We use
\[
    r= \lceil n \cdot (1 + \epsilon) / N \rceil
\]
frontyard mini-structures.  A fixed top-level hash function
\(g:U\to\{1,2,\ldots,r\}\) assigns each key to one mini-structure. 

The total number of frontyard slots is
\[
    rM=rN+O(rN^{4/5}) \le n \cdot (1 + \epsilon) + N + O(n / N^{1/5}) \le n + O(\epsilon n),
\]
where the final step uses that $N^{1/5} \ge \Omega(\epsilon^{-1})$ and that $N \le O(\epsilon n)$, 
both of which follow from the facts that $N = \epsilon^{-5}$ and $\epsilon \ge n^{-\alpha}$ for some sufficiently small $\alpha > 0$.

\paragraph{The backyard. }
The remaining slots are allocated to a backyard of size
\[
    m_{\mathrm B}= \epsilon n.
\]

The backyard is implemented using the deamortized cuckoo hashing construction
of Arbitman, Naor, and Segev~\cite{deamortizedcuckoo}. 
We use the following black-box form of their result.

\begin{lemma}[Deamortized Cuckoo Hashing \cite{deamortizedcuckoo}]
\label{lem:backyard-dictionary}
Fix constants \(D>0\) and \(c>0\).  There is an open-addressed hash table using 
\(m\) slots and $O(\log m)$ machine words of metadata that offers the following guarantee.  For any oblivious sequence of
at most \(m^c\) insertions, deletions, and queries, if the number of live keys
is always at most \(m/3\), then, with probability at least \(1-m^{-D}\), every
operation succeeds and every operation takes worst-case \(O(1)\) time.
\end{lemma}

\paragraph{Queries, insertions, and deletions. } We now describe how the frontyard and 
backyard are used to implement the operations.

Queries can simply query the backyard (in time $O(1)$) and one frontyard mini-structure (in time $O(\log\epsilon^{-1})$).
Insertions and deletions, on the other hand, follow a slightly more intricate protocol. The basic idea is to set a parameter
\[
    L=C_{\mathrm L}\log N=\Theta(\log\epsilon^{-1}),
\]
where \(C_{\mathrm L}\) is a sufficiently large constant, and then to use the following rule: By default, 
we perform operations in the appropriate frontyard mini-structure, but whenever an operation takes
time more than $L$, we truncate it early, and if necessary, we move elements to the backyard to preserve correctness.

More concretely, we insert a key $x$ as follows.  We begin the insertion in
mini-structure \(g(x)\) and stop after time \(L\).  If the
insertion completes within time \(L\), then the operation is finished.  If the
local mini-structure reports failure, or if the insertion has not completed
after \(L\) time, then we say that the operation \defn{timed out}.
In this case, we stop the insertion and put the one currently
unplaced key into the backyard. (The unplaced key may not be the key
being inserted if the insertion has performed part of an eviction chain.) 

Note that, if the insertion times out, then there is no need to roll back any partial
 eviction chain that occurred.  At every moment of
the chain, all placed keys are still stored in one of their legal frontyard
locations; only the current key in the eviction chain is unplaced -- that is the key
that gets moved into the backyard. 

To delete a key \(x\), we first locate it via a query.  If the key is in the backyard
we delete it in $O(1)$ time. Otherwise, we delete it from its frontyard mini-structure
using the deletion protocol for the mini-structure. 
As before, if the procedure takes time more than $L$, we say that it \defn{times out}, and we 
end the pullback procedure early. Note that this does not violate correctness 
(every remaining key is still stored in one of its legal locations). This may,
 however, violate the Uniform Overflow Invariant
 in the mini-structure $g(x)$, which may lead to future insertions incurring failures 
 and then using the backyard. 

Finally, time is divided into \defn{epochs} of length \(R=n\) insertions/deletions.  At the beginning of each epoch,
we rebuild the entire structure from the current live set.  This rebuild is the
only source of amortization, and is performed in the most straightforward way possible: It moves the keys in the data structure to a temporary secondary table, it sorts them in linear time into some arbitrary canonical order, and then it reinserts them from scratch into the original hash table.

\subsection{Bounding the Number of ``Bad'' Mini-Structures}

The main difficulty in the analysis is to bound the size of the backyard, at any given moment, with high probability. To do this, we will first focus on bounding the number of ``bad'' mini-structures (namely, those that may send keys to the backyard). 

Fix an epoch.  Let \(A\) be the set of keys that are live at the beginning of
the epoch or inserted at some point during the epoch.  Since at most \(n\) keys
are live at the beginning of the epoch and the epoch contains \(n\) insertions/deletions, 
\(|A|\le 2n\).  For each mini-structure \(a\), write
\[
    A_a=\{x\in A:g(x)=a\}.
\]

We call a mini-structure \defn{bad} during the epoch if any of the following three
events occurs.

\begin{enumerate}

    \item \emph{Traffic bad:} more than \(N^3\) operations during the epoch
    involve keys assigned to the mini-structure.  Here we count both the
    insertions performed during the global rebuild and the insertions/deletions during
    the epoch.

    \item \emph{Capacity bad:} at some time during the epoch, more than \(N\)
    live keys are assigned to the mini-structure.

    \item \emph{Locally bad:} before the mini-structure has already become
    capacity bad or traffic bad, one of its insertions or deletions times out. 
\end{enumerate}

Every key that is ever physically placed in the backyard during the epoch is
assigned to a bad mini-structure.  Indeed, a key enters the backyard only when a
frontyard insertion attempt fails or times out.  If the corresponding
mini-structure was already capacity bad or traffic bad, then it is bad by
definition. Otherwise, that failed or timed-out insertion attempt makes it
locally bad.  Deletion timeouts can also make a mini-structure locally bad, but
they do not themselves add keys to the backyard.\footnote{Nonetheless, if a deletion timeout
occurs then the mini-structure violates its history-independent guarantee, 
and we can no longer reason about future insertions. This is why we include
deletion timeouts in the definition of locally bad.}

We now prove three lemmas that bound the number of bad mini-structures. 
During the fixed epoch that we are considering, let $\mathcal B_{\mathrm{traf}}$ be the set of traffic-bad mini-structures,
$\mathcal B_{\mathrm{cap}}$ be the set of capacity-bad mini-structures, and
$\mathcal B_{\mathrm{loc}}$ be the set of locally bad mini-structures.

\begin{lemma}[Few Traffic-Bad Mini-Structures]
\label{lem:few-traffic-bad}
$|\mathcal B_{\mathrm{traf}}| \le 2n/N^3$.
\end{lemma}
\begin{proof}
The total number of local operations during an epoch, including the rebuild
insertions, is at most \(2n\).  Each traffic-bad mini-structure receives more
than $N^3$ such operations. Therefore the lemma follows. 
\end{proof}

\begin{lemma}[Few Capacity-Bad Mini-Structures]
\label{lem:few-capacity-bad}
With probability at least \(1-\exp(-n^{\Omega(1)})\),
\[
    |\mathcal B_{\mathrm{cap}}|\le n / N^3.
\]
\end{lemma}

\begin{proof}
Recall that there are $r = (1 + \Theta(\epsilon)) n /N$ mini-structures in total.
At any given moment,
for any fixed mini-structure \(a\), the number of live keys $S_a$ assigned to \(a\) has
mean at most
\[
    \frac{n}{r}\le (1-\Omega(\epsilon))N.
\]
By a Chernoff bound,
\begin{equation}\label{eq:capacity-bad-prob}
    \Pr[|S_a|\ge N]\le \exp(-\Omega(\epsilon^2N)) \le \exp(-\Omega(\epsilon^{-1})).
\end{equation}
Thus the expected number of mini-structures that are already over capacity at
the beginning of the epoch is at most
\(r\exp(-\Omega(\epsilon^{-1})) \le n \exp(-\Omega(\epsilon^{-1}))\).

After the beginning of the epoch, a mini-structure can become capacity bad only
when an insertion increases its assigned live load from \(N\) to \(N+1\). Consider
an insertion of a key $x$ that maps to mini-structure $g(x)$. Prior 
to the insertion, the probability of the mini-structure $g(x)$ having at least $N$ live
keys is, by \eqref{eq:capacity-bad-prob}, at most $\exp(-\Omega(\epsilon^{-1}))$.
The probability that the insertion causes a capacity violation is therefore
also at most $\exp(-\Omega(\epsilon^{-1}))$. 

Now view \(|\mathcal B_{\mathrm{cap}}|\) as a function of the values
$g(x)$ for $x \in A$.  Changing $g(x)$ for a single $x \in A$ can change
 the capacity-bad status of at most two mini-structures: the key's old
mini-structure and its new mini-structure.  Thus the function is
\(2\)-Lipschitz.  McDiarmid's inequality gives
\[
    \Pr\left[
        |\mathcal B_{\mathrm{cap}}|> \E[|\mathcal B_{\mathrm{cap}}|] + k \sqrt{n}
    \right]
    \le
    \exp\left(
        -\Omega\left(k^2\right)
    \right).
\]
Plugging in $k = \sqrt{n}/(2N^3)$ gives 
\[
    \Pr\left[
        |\mathcal B_{\mathrm{cap}}|> n e^{-\Omega(\epsilon^{-1})} + n/(2N^3)
    \right]
    \le
    \exp\left(
        -\Omega\left(n/N^6\right)
    \right).
\]
Since $\epsilon^{-1}  \le n^\alpha$ for a sufficiently small $\alpha > 0$, and $N = \epsilon^{-5}$, we have
$n e^{-\Omega(\epsilon^{-1})} + n/(2N^3) \le n / N^3$ and $\exp(
        -\Omega(n/N^6)) = \exp(
        -n^{\Omega(1)})$, which completes the proof.
\end{proof}

\begin{lemma}[Few Locally Bad Mini-Structures]
\label{lem:few-locally-bad}
With probability at least \(1-\exp(-n^{\Omega(1)})\), 
\[
    |\mathcal B_{\mathrm{loc}}|\le \frac{n}{N^3}.
\]
\end{lemma}

\begin{proof}
Condition on the hash function \(g\), and also condition on the sets
of capacity-bad and traffic-bad mini-structures.  Consider any remaining
mini-structure \(a\).  During the epoch it has at most \(N\) assigned live keys
at all times and receives at most \(H=N^3=N^{O(1)}\) local operations,
including the rebuild insertions.

Until the first locally bad event occurs in \(a\), the local operation sequence
seen by \(a\) is fixed by the global oblivious update sequence and by \(g\), and
is independent of the local randomness used inside \(a\).  Therefore
Theorem~\ref{thm:dynamic-whi} bounds the probability that \(a\) becomes locally bad during the epoch
by at most, say, \(N^{-3}/2\).

Conditioned on \(g\), the local randomness used by different mini-structures is
independent.  Hence \(|\mathcal B_{\mathrm{loc}}|\) is dominated by a sum of
independent indicator random variables with total mean at most 
\(N^{-3}r/2 \le N^{-3} n/2\).
A Chernoff bound gives
\[
    \Pr\left[
        |\mathcal B_{\mathrm{loc}}|> N^{-3} n
        \,\middle|\, g
    \right]
    \le
    \exp(-\Omega(N^{-3} n)) \le \exp(-n^{\Omega(1)}).
\]
\end{proof}

Combining Lemmas~\ref{lem:few-traffic-bad},
\ref{lem:few-capacity-bad}, and \ref{lem:few-locally-bad}, we obtain the
following.

\begin{corollary}[Few Bad Mini-Structures]
\label{cor:few-bad-minis}
For a given epoch, we have with high probability in $n$
that the number of bad mini-structures is at most $4n / N^3$. 
\end{corollary}

\subsection{Bounding the Number of Backyard Keys}

Having bounded the number of bad mini-structures (within a given epoch), we bound the size of the backyard by bounding the total number of keys that (during that epoch) hash to those bad mini-structures. To do this, we make use of the following lemma: 
\begin{lemma}[Largest Mini-Structures]
\label{lem:largest-minis}
Let \(s\le 2n\) keys be assigned independently and uniformly to
\(r \ge n/N\) mini-structures. Let \(Y_a\) be the load of mini-structure \(a\),
and let \(Z_k\) be the sum of the \(k\) largest values among
\(Y_1,\ldots,Y_r\). Then,
\[
    Z_k\le 4kN + n e^{-\Omega(N)} + \tilde{O}(\sqrt{n})
\]
with high probability in $n$. 
\end{lemma}
\begin{proof}
The expected number $Y_a$ of keys, out of the $s \le 2n$, that are assigned to a 
given mini-structure $a$ is $s/r \le 2N$. By a Chernoff bound,
$$\Pr[Y_a \ge 4N] = e^{-\Omega(N)}.$$
It follows that, for any key $x$, the probability that $x$ hashes to a 
mini-structure $a$ satisfying $Y_a \ge 4N + 1$ (i.e., ignoring $x$, we have $Y_a \ge 4N$)
is at most $e^{-\Omega(N)}$. Thus, if we define $Y$ to be the total number of 
keys that hash to mini-structures with load at least $4N + 1$, we have 
$$\E[Y] \le 2n e^{-\Omega(N)}.$$
Since $Z_k \le Y + 4kN$, it follows that
$$\E[Z_k] \le 2n e^{-\Omega(N)} + 4kN.$$

Finally, observe that $Z_k$ can be written as a function of the hashes $g(x)$ 
of each of the $s \le 2n$ keys $x$. Changing a given hash value changes $Z_k$ by at most
$1$, so we can apply McDiarmid's inequality to conclude that, with high probability in $n$,
$Z_k \le \E[Z_k] + \tilde{O}(\sqrt{n})$. This completes the proof.
\end{proof}

Putting the pieces together, we can now argue that the backyard does not contain too many elements (and therefore is small enough for Lemma~\ref{lem:backyard-dictionary} to apply). 

\begin{corollary}[Backyard Size]
\label{cor:backyard-size}
With high probability in $n$, the number of keys in the backyard is never more than
$O(n / N^2)$.
\end{corollary}
\begin{proof}
Since there are only $\poly(n)$ total epochs, it suffices to prove the result during a given epoch. 
By Corollary~\ref{cor:few-bad-minis}, the number of bad mini-structures during the epoch is at most $4n / N^3$ with high probability in $n$. 
By Lemma~\ref{lem:largest-minis}, the number of keys in the backyard is at most $O(n / N^2) + n e^{-\Omega(N)} + \tilde{O}(\sqrt{n})$. Recalling that $N = \epsilon^{-5}$ and that $\epsilon \ge n^{-\alpha}$ for a sufficiently small $\alpha > 0$, we have that $\tilde{O}(\sqrt{n}) = O(n / N^2)$. Therefore the number of keys in the backyard is at most $O(n / N^2)$ with high probability in $n$.
\end{proof}

\subsection{Proof of Theorem~\ref{thm:epsilon-load-factor}}

Finally, we are now in a position to prove Theorem~\ref{thm:epsilon-load-factor}. 
The frontyard uses at most \(n+O(\epsilon n)\) slots, and the backyard uses at
most \(\epsilon n\) slots.  Thus the total number of slots is at most
$n + O(\epsilon n)$. The only metadata is the $O(\log n)$ machine words of metadata
from Lemma \ref{lem:backyard-dictionary}, and the only temporary space is the $O(n)$
machine words of space used during rebuilds. Thus the space bounds
in Theorem~\ref{thm:epsilon-load-factor} are satisfied.

Next we turn to proving time bounds. Queries check the backyard and one frontyard mini-structure.  The backyard
query costs \(O(1)\), and the mini-structure query costs
\(O(\log N)=O(\log\epsilon^{-1})\).  Thus queries take worst-case time
\(O(\log\epsilon^{-1})\).

Consider an insertion or deletion that is not part of a global rebuild.  If
the operation touches the frontyard, the frontyard procedure is stopped after
\(L=\Theta(\log N)=\Theta(\log\epsilon^{-1})\) time.  Thus the
frontyard work of any non-rebuild update is worst-case
\(O(\log\epsilon^{-1})\).  Any backyard work costs worst-case \(O(1)\) on the
success event of Lemma~\ref{lem:backyard-dictionary}.

Every \(n\) updates, we rebuild the entire data structure.  The rebuild cost is
\(O(n\log N)=O(n\log\epsilon^{-1})\), so rebuilds contribute
\(O(\log\epsilon^{-1})\) amortized time per update.  Therefore insertions and
deletions take expected amortized time \(O(\log\epsilon^{-1})\).

Finally, and most importantly, we bound the probability of the backyard failing.
By Corollary~\ref{cor:backyard-size}, 
with high probability in $n$, the
backyard never contains more than \(O(n / N^2)\) keys. Recalling that
$N =\epsilon^{-5}$, and that $\epsilon \le \epsilon_0$ for a sufficiently small $\epsilon_0 > 0$,
it follows that the backyard never contains more than $\epsilon n / 3$ keys.
Conditioned on this event,
Lemma~\ref{lem:backyard-dictionary} implies that the backyard succeeds
throughout the polynomial-length update sequence with high probability in $n$. 
Since the only way for the overall data structure to fail is for the backyard to fail,
the overall data structure succeeds with high probability in $n$. This completes the proof.

\section{Deamortizing the Rebuilds}
\label{sec:deamortization}

We now explain how to remove the global rebuilds from the construction of
Section~\ref{sec:all-load-factors}.  The goal is to obtain high-probability
worst-case bounds for insertions and deletions, rather than expected amortized
bounds.

The only reason that the construction in Section~\ref{sec:all-load-factors} 
does not already achieve high-probability worst-case time bounds is because
of the time spent in rebuilds. The task
in this section is to spread the work of each rebuild over the operations
of that epoch. This will give us the following theorem:

\begin{theorem}[Deamortized load factor \(1-\epsilon\)]
\label{thm:epsilon-load-factor-deamortized}
Let $n^{-\alpha} \le \epsilon \le \epsilon_0$ for some sufficiently small constants $\alpha, \epsilon_0 > 0$.
There is an open-addressed hash table using $n + O(\epsilon n)$ slots and
$O(\log n)$ machine words of metadata, that offers the following
guarantee for any oblivious polynomial-length insertion/deletion sequence in which the number of live keys
never exceeds $n$: With high probability in $n$, every operation succeeds and every
insertion and deletion takes worst-case time $O(\log\epsilon^{-1})$; additionally, queries take worst-case time $O(\log\epsilon^{-1})$.
\end{theorem}

Note that the assumption $\epsilon \le \epsilon_0$ is WLOG since replacing $\epsilon > \epsilon_0$ with $\epsilon' = \min(\epsilon, \epsilon_0)$ does not change the asymptotic guarantees of the theorem. Furthermore, the requirement $\epsilon \ge n^{-\alpha}$ can be substituted with $\epsilon \ge n^{-1/4}$ since, for $\epsilon^{-1} \in [n^\alpha, n^{1/4}]$, we can use Theorem \ref{thm:general-capacities} instead of Theorem \ref{thm:epsilon-load-factor-deamortized}.
Therefore, as an immediate corollary, we have the main theorem of the paper, Theorem \ref{thm:main}.

\maintheorem*

The rest of the section is spent proving Theorem \ref{thm:epsilon-load-factor-deamortized}.

\paragraph{Warmup: Performing rebuilds in-place.}
We begin with a simpler task: implementing the rebuild step from the previous section,
but in-place, so that no temporary memory is needed. In addition, we will implement
the rebuild in such a way that, if a \emph{query} were to be performed at any time
during the rebuild, it would still be correctly answered.

The rebuild procedure will expand the data structure to include two additional sets of slots:
\begin{enumerate}
\item \textbf{A mirror workspace:} This is an array of $O(N)$ slots that perfectly mirrors a mini-structure. Any element $y$ that is allowed to go in some slot $s$ of some mini-structure $g(y)$ is also allowed to
go to slot $s$ in the mirror workspace (queries therefore check not just $g(y)$ but also the mirror workspace).
The mirror workspace will be used as a temporary scratch space when we rebuild a mini-structure. The mirror workspace is not temporary memory -- it is a permanent part of the array, and there is a single mirror workspace for the entire data structure.
\item \textbf{A second backyard: } Rather than one backyard, we will now have two, $B_1$ and $B_2$. This is so that, as we perform a rebuild, which attempts to clear out one of the backyards, say, $B_1$, we can put any elements that incur
failures \emph{during the rebuilds} (i.e., their mini-structure becomes capacity bad) into the other backyard, say $B_2$. During the rebuild, we refer to the original backyard as the \emph{old backyard} and the new backyard as the \emph{new backyard}. At the end of the rebuild, the old backyard is empty, and the new backyard becomes the backyard in-use for the next epoch -- the two backyards then swap roles during the next rebuild, and so on. 
\end{enumerate}
Note that these additions to the data structure consist of $O(\epsilon n)$ slots, so they do not change the overall space bound of $n + O(\epsilon n)$ slots.

The basic task of a rebuild is to move any keys from backyard
back to the appropriate frontyard mini-structure (if we can). 
Naively, we can do this by scanning the backyard and attempting
to insert each key into the appropriate frontyard mini-structure.
However, this strategy would fail to restore correctness for 
mini-structures that have incurred failures -- these mini-structures
need to be entirely rebuilt so that they once again satisfy the
 Uniform Overflow Invariant.

Thus, to perform a rebuild, we take a more heavy-handed strategy. 
As in the naive strategy, we scan the backyard and attempt to move 
each key $x$ that we encounter back to its frontyard mini-structure $g(x)$.
But, rather than using the usual insertion rule, we instead rebuild the
entire mini-structure $g(x)$ as follows:
\begin{enumerate}
\item We copy the elements of the mini-structure $g(x)$ to the mirror workspace.
\item We then scan the mirror workspace, and insert each element $y$ in the mirror workspace
back into mini-structure $g(x)$, \emph{but using the strongly history-independent protocol from Theorem \ref{thm:dynamic-whi} (really Theorem \ref{thm:general-capacities}).}
\item Finally, we also insert $x$ into $g(x)$ using the same strongly history-independent protocol.
\item If any of these insertions fail, then we declare mini-structure $g(x)$ to be capacity bad,
and we send the remaining elements in the mirror workspace to the new backyard.
\end{enumerate}

Finally, a bit of extra care is needed to make sure that \emph{every} failed mini-structure
gets rebuilt during the rebuild. To ensure this, whenever a deletion times out during an epoch,
move one element (any element in the mini-structure that timed out) to the backyard. This ensures
there is a ``token element'' in the backyard for every mini-structure that has incurred a failure.

It is worth taking a moment to explain why we use the strongly history-independent protocol here, rather than the more efficient weakly history-independent one. One oddity of the
construction above is that, as it inserts elements from the mirror workspace back into the mini-structure $g(x)$, the \emph{order} in which they are inserted is not necessarily oblivious to their hashes in $g(x)$. Indeed, the order is the same as the order in which they appear in the mirror workspace, which is the same as the order they formerly appeared in $g(x)$. What is nice about the strongly history-independent protocol is that, unless it incurs failure, it guarantees worst-case $O(\log^2 N) = O(\log^2 \epsilon^{-1})$ time per operation no matter what order the elements are inserted in. Thus, if the mini-structure does not incur a failure during the rebuild, then the rebuild takes worst-case time $O(N \log^2 \epsilon^{-1})$ and results in a mini-structure that satisfies the Uniform Overflow Invariant.

This concludes the description of the in-place rebuild procedure. Note that the total time taken by the procedure is at most $O(K \cdot N \log^2 \epsilon^{-1})$, where $K$ is the number of elements in the backyard. Supposing $K = O(n / N^2)$, this is $O(n)$. 

\paragraph{The deamortized data structure.}
To obtain a data structure with (high-probability) worst-case time bounds, we make one additional modification:
\begin{enumerate}
\item We split the data structure into equal-size (up to $\pm 1$) chunks of size $R = \Theta(n^{1/3})$. 
\item Each element $x$ hashes to a random chunk $r(x)$.
\item Each chunk then uses the construction from the previous section, with its own backyard. 
\item For every $R$ operations that occur on the data-structure as a whole, we perform a rebuild on
one of the chunks (using the in-place procedure described above), rotating through the chunks so that
each chunk is rebuilt every $O(n)$ operations. The rebuild is performed using $O(\log N)$ time per operation,
so as long as the rebuild can complete in time $O(R)$, then it can complete within its scheduled window.
\end{enumerate}

When a rebuild is performed, queries to that chunk
proceed as usual (rebuilds never break queries); insertions directly use the new backyard (they 
skip the frontyard mini-structure); and deletions either use the appropriate backyard (if the key is in a backyard)
or delete the key from the appropriate frontyard mini-structure (or mirror workspace) without doing anything else. 

This approach guarantees that insertions/deletions/queries to the chunk being rebuilt take worst-case time $O(\log \epsilon^{-1})$. On the other hand,
this also means that, when a key $x$ that hashes to the chunk being rebuilt is inserted/deleted,
we do nothing to try to maintain the Uniform Overflow Invariant for that key $x$'s frontyard mini-structure $g(x)$. Declare such a frontyard mini-structure to be \defn{data-race bad}. We will, as part of the analysis, argue that there are very few data-race bad mini-structures during the rebuild of any given chunk.

\paragraph{Avoiding a blowup in backyard metadata.}
A naive application of the construction above would seem to multiply the $O(\log n)$ metadata used by Lemma \ref{lem:backyard-dictionary} (i.e., the backyard) by the number of chunks. However, so long as the backyard dictionaries all stay under capacity, then they can all share a single set of $O(\log n)$ machine words of metadata. This metadata is just used to keep track of the queue/stash/cycle-detection-mechanism in the construction of \cite{deamortizedcuckoo}. For two of these sources of metadata, the queue size and the cycle-detection-mechanism size, the bound of $O(\log n)$ proven in \cite{deamortizedcuckoo} applies equally well to the setting where, as here, the hash table is broken into polynomial-size chunks.

Thus the only minor point of trickiness is the number of ``stash'' elements, i.e., elements that have incurred a genuine failure (they are in a bicyclic component) in one of the backyard cuckoo hash tables during the most recent $O(n)$ operations. The standard analysis (e.g., Theorem 2.1 of \cite{stashcuckoo} or Lemma 4.5 of \cite{deamortizedcuckoo}) shows that the number of such elements per chunk is $O(1)$ with high probability in $n$; moreover, the \emph{expected} number of elements per chunk (based only on the randomness in that chunk) is, by the standard analysis of cuckoo hashing, $O(1/R^2) = O(n^{-2/3})$. Summing over the $O(n/R) = O(n^{2/3})$ chunks, the total number of stash elements is a sum of independent $[0, O(1)]$-random variables with total mean $O(1)$, and is therefore $O(\log n)$ with high probability in $n$. 

Thus, even though we now have many backyards, they can collectively use just $O(\log n)$ machine words of metadata.

\paragraph{Analysis. } To analyze the deamortized data structure, the main technical challenge (just as in the previous section) is to argue that no backyard ever becomes overly full. 

Focus on a given chunk $Q$, and define an epoch for that chunk to be the period of time starting from the last rebuild (and including that rebuild) up until the start of the next rebuild (not including that rebuild). Define the \defn{operation sequence} for the chunk, during the epoch, to be the operations on keys that hash to that chunk.

Unlike in the previous section, the epoch for chunk $Q$ is not guaranteed to have exactly $|Q| = R$ operations. The next lemma splits elements into \defn{heavy hitters} and \defn{light hitters}, and argues that the number of operations on light hitters is $O(R)$ with high probability in $n$.

\begin{lemma}
\label{lem:heavy-hitters}
Let $Q$ be a chunk of size $R = \Theta(n^{1/3})$. Let $H$ be the set of keys that hash to $Q$ and experience at least $N^{10}$ insertions/deletions during the epoch (the \emph{heavy hitters}), and let $L$ be the set of keys that hash to $Q$ and experience fewer than $N^{10}$ insertions/deletions during the epoch (the \emph{light hitters}). Then, with high probability in $n$, 
\begin{itemize}
\item The number of keys in $H$ is $O(R / N^{10})$;
\item The number of insertions/deletions on keys in $L$, during the epoch, is $O(R)$.
\end{itemize}
\end{lemma}
\begin{proof}
During the epoch, at most $O(n)$ insertions/deletions occur overall, so there are at most $O(n / N^{10})$ elements that incur at least $N^{10}$ insertions/deletions. Each of these hashes to $Q$ with probability $O(R/n)$, so $|H|$ is a sum of iid 0-1 random variables with total mean $O(R / N^{10})$. Since $R / N^{10} \ge \sqrt{R} \ge \omega(\log n)$ (by the assumption that $\epsilon^{-1}$ is at most a small polynomial of $n$), we have by a Chernoff bound that $|H|$ is $O(R / N^{10})$ with probability $1 - e^{-\Omega(R / N^{10})} \ge 1 - 1/\poly(n)$. 

On the other hand, the number of insertions/deletions on keys in $L$, during the epoch, is a sum of independent (non-iid) $[0, N^{10}]$-random variables with total mean $O(R)$ (one random variable for each key that is accessed fewer than $N^{10}$ times during the epoch). By assumption, $N^{10} \le R^{1/3}$, so a standard Chernoff bound gives that this number is $O(R)$ with probability at least $1 - e^{-\Omega(R / N^{10})} \ge 1 - 1/\poly(n)$.
\end{proof}

Define traffic-bad mini-structures, capacity-bad mini-structures, and locally-bad mini-structures as in the previous section. Also define data-race bad mini-structures as above. Any mini-structure that sends any keys to the backyard during the epoch is guaranteed to be bad in one of these ways. Note that any heavy hitter (as defined above) is guaranteed to be in a traffic-bad mini-structure. 

The bounds on the number of capacity-bad mini-structures and locally-bad mini-structures follow exactly as in the previous section: with high probability in $n$, there are at most $O(R / N^3)$ such mini-structures. 

The bound on the number of traffic-bad mini-structures is also straightforward (assuming the high-probability event from Lemma \ref{lem:heavy-hitters}): Any traffic-bad mini-structure either has at least one key from $H$ that hashes to it (there are at most $O(R / N^{10})$ such keys); or receives at least $N^3$ operations on keys from $L$ (since the total number of operations on keys in $L$ is $O(R)$ by Lemma \ref{lem:heavy-hitters}, there can be at most $O(R / N^3)$ such mini-structures).  Therefore, with high probability in $n$, there are at most $O(R / N^3)$ traffic-bad mini-structures. 

Finally, we can bound the number of data-race bad mini-structures as follows:
\begin{lemma}
\label{lem:data-race-bad-mini-structures}
With high probability in $n$, there are at most $O(1)$ data-race bad mini-structures.
\end{lemma}
\begin{proof}
Recall that the rebuild spans $O(R) = O(n^{1/3})$ operations (in the whole data structure). Let $S$ be the set of keys accessed by these operations. Each key in $S$ independently has probability $O(R/n) = O(n^{-2/3})$ of hashing to chunk $Q$. So the number of such keys that hash to $Q$ is at most a sum of iid 0-1 random variables with mean $O(n^{-1/3})$. By a Chernoff bound, such a sum is $O(1)$ with high probability in $n$. The number of data-race bad mini-structures is therefore also $O(1)$ with high probability in $n$. 
\end{proof}

Putting the bounds together, the total number of bad mini-structures in chunk $Q$, during a given epoch, is at most $O(R / N^3)$ with high probability in $n$. By the same argument as in Lemma \ref{lem:largest-minis}, we can conclude that, of the $O(R)$ keys in $L$, at most $O(R/N^2)$ hash to bad mini-structures (with high probability in $n$). 
Combining these with the $O(R / N^{10})$ keys in $H$, the total number of keys that hash to bad mini-structures in $Q$
during the epoch, and therefore the size of the backyard, is at most $O(R / N^2)$ with high probability in $n$. 

Applying a union bound over all chunks and all epochs, we can conclude that, with high probability in $n$, no backyard ever has more than $O(R / N^2)$ keys in it. This means that each backyard meets the requirements for Lemma \ref{lem:backyard-dictionary}, offering $O(1)$ time operations with high probability in $R$ (and therefore also in $n$). This $O(R/N^2)$ bound also means that the time per rebuild is $O(R)$, meaning that each rebuild can successfully complete within its scheduled window. 

By the design, the new data structure offers worst-case $O(\log \epsilon^{-1})$ time operations (with high probability in $n$), so Theorem \ref{thm:epsilon-load-factor-deamortized} follows.

\section{A Lower Bound on Insertion/Deletion Time}
\label{sec:lower}

In this section, we prove that any open-addressed hash table with $O(\log \epsilon^{-1})$ worst-case query time must also move $\tilde{\Omega}(\log \epsilon^{-1})$ keys per insertion/deletion, on average. To reason about an arbitrary open-addressed hash table, we will use the following abstraction.  Each key
\(x\) has a probe set \(L(x)\subseteq [m]\), determined by its hash value, and
the table may store \(x\) only in cells of \(L(x)\).  A worst-case successful
query bound of \(d\) means that \(|L(x)|\le d\) for every key \(x\).

Throughout the section, all logarithms are base two.  We suppress floors and
ceilings, as they do not affect the asymptotics.  A move means relocating a key
that is already present in the table; the initial placement of a newly inserted
key is not counted.

\begin{theorem}
\label{thm:open-addressing-movement-lower-bound}
Let an open-addressed hash table have \(m\) cells and load factor
\(1-\varepsilon\), where \(0<\varepsilon<\varepsilon_0\) for a sufficiently
small constant \(\varepsilon_0\).  Suppose every key is stored in one of at most
\(d=O(\log \varepsilon^{-1})\) possible cells.  Starting from any configuration
with \(n=(1-\varepsilon)m\) keys, delete \(r=10\varepsilon m\) uniformly random
stored keys and then insert \(r\) new keys.  Then the expected total number of
moves during this sequence is
\[
        \Omega\!\left(
            \frac{\varepsilon m \log \varepsilon^{-1}}
                 {\log\log \varepsilon^{-1}}
        \right).
\]
Consequently, the expected number of moves per update is
\[
        \Omega\!\left(
            \frac{\log \varepsilon^{-1}}
                 {\log\log \varepsilon^{-1}}
        \right).
\]
\end{theorem}

\begin{proof}
The proof is a one-way communication argument.  The encoder knows the deletion
set \(D\), runs the update sequence, and sends a binary message to a decoder who
does not know \(D\).  The encoder and decoder are both given, for free, the
initial table layout, the probe sets of all keys, and the sequence of inserted
keys.  If the data structure is randomized, we also give its random tape to both
parties; equivalently, we prove the bound after conditioning on the tape and
then average over it.

At a high level, the message will describe certain displacement paths in the final table.  These
paths allow the decoder to recover many members of \(D\).  If few keys are
moved, we will be able to guarantee that the message encoding the paths is 
very short (roughly $O(\log d)$ times the number of moves bits).  On the other hand, any message that reveals
many keys of $D$ must have large entropy.  Comparing these two
bounds is what ultimately gives the theorem.

Label the \(n=(1-\varepsilon)m\) initially stored keys by \([n]\), and let
\[
        D\subseteq [n],
        \qquad |D|=r=10\varepsilon m,
\]
be the uniformly random deletion set.  We write \(H(\cdot)\) for Shannon entropy
and \(I(\cdot;\cdot)\) for mutual information.  Since \(D\) is a uniform
\(r\)-subset of \([n]\),
\[
        H(D)=\log {n\choose r}.
\]

Let \(E\) be the set of initially empty cells, so \(|E|=\varepsilon m\).  For an
old key \(x\) (i.e., a key that was stored in the initial configuration), let \(c_0(x)\) denote its initial cell.  For any key \(z\) present
in the final table, let \(c_1(z)\) denote its final cell.

We now define a collection of displacement paths by comparing the initial and
final layouts. 

For each inserted key $y$, we will have one path starting from $y$. A path will specify a sequence of elements $z_1, z_2, \ldots, z_k$ and positions $p_1, p_2, \ldots, p_k$ such that $z_1 = y$; such that $z_i$, for $i \ge 2$, is an old key that was initially in position $p_{i - 1}$ and is now (in the final table) in position $p_i$; and such that $p_k$ is a slot that was either empty in the initial table or that contained a key in $D$. 

Starting from $y = z_1$, we define $p_1 = c_1(y)$, and we construct the displacement path incrementally as follows. Given a point $(z_i, p_i)$ in the path: If $p_i \in E$ or $p_i = c_0(x)$ for some $x \in D$, then the path ends at $z_i$; and otherwise, we define $z_{i + 1}$ to be the old key with $c_0(z_{i + 1}) = p_i$, and we define $p_{i + 1}$ to be $c_1(z_{i + 1})$. By continuing this construction until it completes, we get the full displacement path starting from $y$.

We remark that the paths above are non-cyclic and non-overlapping (both in terms of the keys and the positions they visit). The key observation is that paths are reversible: the predecessor of $z_{i+1}$ on its path is the unique key $z$ with $c_1(z) = c_0(z_{i+1})$, so given any point on a path, together with the initial and final configurations of the hash table, we can rebuild the path backwards. In particular, every key has at most one predecessor, and the inserted key $z_1 = y$ has none. If a path revisited a key, then the first revisited key would either be $y$ or have two distinct predecessors, both of which are impossible. And if two paths shared a key, then rebuilding backwards from that key would lead to a single starting key, so the paths would be the same. Therefore the paths are both non-cyclic and non-overlapping. 

Thus the \(r\) inserted keys define \(r\) disjoint paths, each terminating
either at an initially empty cell or at the initial cell of a deleted key.  The
endpoints are distinct.  Since only \(|E|=\varepsilon m\) endpoints can be
initially empty cells, at least
\[
        r-|E| = 9\varepsilon m
\]
paths terminate at cells that initially contained deleted keys.  Hence the paths
identify a subset
\[
        S\subseteq D
\]
of size at least
\[
        s=9\varepsilon m.
\]

Let \(P\) be the number of old surviving keys (keys present both initially and in the final configuration) that appear on these displacement
paths.  Every such key must move between the initial and final configuration, so the total number of moves $T$ must satisfy 
\[
        T\ge P.
\]

The encoder now sends a displacement transcript \(\Pi\) that encodes the paths as follows. The paths are listed
in the order of the insertions.  For each key \(z\) appearing on a path (including the element $y$ being inserted), the
transcript records the index of \(c_1(z)\) within \(L(z)\), together with one bit
indicating whether the path stops at that cell.  Thus each step in the path costs
\[
        b := \lceil \log d\rceil + 1 = O(\log(2d))
\]
bits.  The total number of steps is \(r+P\): one step for each inserted key, and
one additional step for each old surviving key on the paths.  Hence
\[
        |\Pi| = b(r+P).
\]

Given \(\Pi\) and the public side information, the decoder reconstructs the
paths as follows. The decoder knows the key $y$ at the start of the path (the key that is inserted), and can recover from the message which position in $L(y)$ is $c_1(y)$. Then, given some $z_i$ in the path and the position $p_i = c_1(z_i)$, the decoder can continue to recover the path as follows: If the stop bit for the path is not set, then the cell $c_1(z_i)$ must have initially contained some old key $q$ with $c_0(q) = p_i$. This key $q$ is $z_{i + 1}$, and using the information in the message, the decoder can recover which position in $L(z_{i + 1})$ is $c_1(z_{i + 1})$. This position is $p_{i + 1}$. Continuing like this, the decoder can recover the entire path starting from $y$. 

Having recovered the paths, the decoder knows that any path ending in a cell that was not initially empty must instead end in a cell whose initial key is in $D$. Therefore \(\Pi\)
determines a set \(S(\Pi)\subseteq D\) with
\[
        |S(\Pi)|\ge s.
\]

We now lower-bound the entropy of the transcript.  The transcript is prefix-free:
each step has fixed length \(b\), the decoder knows that there are exactly \(r\)
paths, and the stop bits tell it exactly where each path ends.  Hence the
standard source-coding inequality gives
\[
        \mathbb{E}[|\Pi|]\ge H(\Pi).
\]
Since the random tape, if any, has been fixed, \(\Pi\) is a deterministic
function of \(D\).  Thus
\[
        H(\Pi)=I(D;\Pi)=H(D)-H(D\mid \Pi).
\]

For any fixed transcript \(\Pi=\pi\), the deletion set \(D\) must contain the
revealed set \(S(\pi)\), and \(|S(\pi)|\ge s\).  Therefore the number of deletion
sets consistent with \(\Pi=\pi\) is at most
\[
        {n-s\choose r-s}.
\]
Consequently,
\[
        H(D\mid \Pi)\le \log {n-s\choose r-s}.
\]
Combining the last three displays,
\[
\begin{aligned}
        \mathbb{E}[|\Pi|]
        &\ge H(\Pi)                                                   \\
        &= H(D)-H(D\mid \Pi)                                          \\
        &\ge \log {n\choose r} - \log {n-s\choose r-s}.
\end{aligned}
\]
Since \(|\Pi|=b(r+P)\), this gives
\[
        \mathbb{E}[r+P]
        \ge
        \frac{\log {n\choose r}-\log {n-s\choose r-s}}{b}.
\]

It remains to estimate the entropy gap.  Using
\(n=(1-\varepsilon)m\), \(r=10\varepsilon m\), and
\(s=9\varepsilon m\),
\[
\begin{aligned}
        \log {n\choose r} - \log {n-s\choose r-s}
        &=
        \log \frac{{n\choose r}}{{n-s\choose r-s}}                 \\
        &=
        \log \prod_{i=0}^{s-1} \frac{n-i}{r-i}.
\end{aligned}
\]
For every \(0\le i<s\),
\[
        n-i \ge n-s = (1-10\varepsilon)m,
        \qquad
        r-i \le r = 10\varepsilon m.
\]
Thus, for sufficiently small \(\varepsilon\),
\[
        \frac{n-i}{r-i}
        \ge
        \frac{1}{20\varepsilon}.
\]
Therefore
\[
        \log {n\choose r} - \log {n-s\choose r-s}
        \ge
        s\log \frac{1}{20\varepsilon}
        =
        \Omega(\varepsilon m\log \varepsilon^{-1}).
\]

Since \(b=O(\log(2d))\), we obtain
\[
        \mathbb{E}[r+P]
        =
        \Omega\!\left(
            \frac{\varepsilon m\log \varepsilon^{-1}}
                 {\log(2d)}
        \right).
\]
Because \(r=10\varepsilon m\),
\[
        \mathbb{E}[P]
        =
        \Omega\!\left(
            \frac{\varepsilon m\log \varepsilon^{-1}}
                 {\log(2d)}
        \right)
        -
        O(\varepsilon m).
\]
Finally, \(d=O(\log \varepsilon^{-1})\), so
\[
        \log(2d)=O(\log\log \varepsilon^{-1}).
\]
For sufficiently small \(\varepsilon\), the first term dominates the
\(O(\varepsilon m)\) subtraction, and therefore
\[
        \mathbb{E}[P]
        =
        \Omega\!\left(
            \frac{\varepsilon m\log \varepsilon^{-1}}
                 {\log\log \varepsilon^{-1}}
        \right).
\]
Since \(T\ge P\), the same lower bound holds for the expected total number of
moves.  Dividing by the \(2r=20\varepsilon m\) updates gives the claimed
per-update bound.
\end{proof}

\section{Proof of the Static Placement Theorem}
\label{sec:static-placement-proof}

In this section, we present the full proof of the Static Placement Theorem. We begin in Subsection \ref{sec:subgaussian} with a review of some (relatively) standard properties of sub-Gaussian random variables. We then show in Subsection \ref{sec:actual-proof} how to leverage these properties to prove the theorem.

\subsection{Properties of  Sub-Gaussian Random Variables}
\label{sec:subgaussian}

A real-valued random variable \(X\) with mean zero is called \emph{sub-Gaussian} with variance \(\sigma^2\) if its moment-generating function satisfies
\[
    \mathbb{E}\left[\exp(\lambda X)\right]
    \le
    \exp\left(\frac{\lambda^2\sigma^2}{2}\right)
    \qquad \text{for all } \lambda\in\mathbb{R}.
\]
In this case, we write
\[
    X \in \mathrm{SG}(\sigma^2).
\]
More generally, for an arbitrary random variable \(Y\), we write
\[
    Y \in \mathrm{SG}(\mu,\sigma^2)
\]
to denote that \(\mathbb{E}[Y]=\mu\) and \(Y-\mu\in \mathrm{SG}(\sigma^2)\); that is, \(Y\) is sub-Gaussian with mean \(\mu\) and variance \(\sigma^2\).

We will routinely use the following standard equivalence between sub-Gaussian
moment-generating-function bounds and Gaussian-type tail bounds.

\begin{restatable}[Sub-Gaussian tails and MGF bounds]{lemma}{subgtailmgf}
\label{lem:subg-tail-mgf}
Let \(Y\) be a random variable with \(\mathbb{E}[Y]=\mu\).
\begin{enumerate}
    \item If \(Y\in \mathrm{SG}(\mu,\sigma^2)\), then for every \(t\ge 0\),
    \[
        \mathbb{P}[|Y-\mu|\ge t]
        \le
        2\exp\left(-\frac{t^2}{2\sigma^2}\right).
    \]
    \item Conversely, there is an absolute constant \(C>0\) such that if there
    exists \(v>0\) with
    \[
        \mathbb{P}[|Y-\mu|\ge t]
        \le
        2\exp\left(-\frac{t^2}{v}\right)
        \qquad \text{for every } t\ge 0,
    \]
    then \(Y\in \mathrm{SG}(\mu,Cv)\).
\end{enumerate}
\end{restatable}

The proof of Lemma~\ref{lem:subg-tail-mgf} is deferred to
Appendix~\ref{app:subg-tail-mgf}. We will also use the following standard closure properties of sub-Gaussian random variables.

\begin{lemma}[Sums of independent sub-Gaussians]
\label{lem:subg-sums}
Let \(X_1,\dots,X_m\) be independent random variables satisfying \(X_i\in
\mathrm{SG}(\mu_i,\sigma_i^2)\). Then
\[
    \sum_{i=1}^m X_i
    \in
    \mathrm{SG}\left(\sum_{i=1}^m \mu_i,\sum_{i=1}^m \sigma_i^2\right).
\]
\end{lemma}

\begin{proof}
By replacing each \(X_i\) with \(X_i-\mu_i\), we may assume without loss of
generality that \(\mu_i=0\) for every \(i\); the general statement then follows
by adding back \(\sum_i \mu_i\), which only shifts the mean. So assume each
\(X_i\in \mathrm{SG}(\sigma_i^2)\). For every \(\lambda\in\mathbb{R}\),
independence gives
\[
    \mathbb{E}\left[\exp\left(\lambda\sum_{i=1}^m X_i\right)\right]
    =
    \prod_{i=1}^m \mathbb{E}\left[e^{\lambda X_i}\right]
    \le
    \prod_{i=1}^m
    \exp\left(\frac{\lambda^2\sigma_i^2}{2}\right)
    =
    \exp\left(
        \frac{\lambda^2}{2}\sum_{i=1}^m \sigma_i^2
    \right).
\]
\end{proof}

\begin{lemma}[Scaling]
\label{lem:subg-scaling}
If \(X\in \mathrm{SG}(\sigma^2)\) and \(a\in\mathbb{R}\), then \(aX \in
\mathrm{SG}(a^2\sigma^2)\). Equivalently, if \(Y\in \mathrm{SG}(\mu,\sigma^2)\),
then \(aY \in \mathrm{SG}(a\mu,a^2\sigma^2)\).
\end{lemma}

\begin{proof}
For every \(\lambda\in\mathbb{R}\),
\[
    \mathbb{E}\left[e^{\lambda aX}\right]
    =
    \mathbb{E}\left[e^{(a\lambda)X}\right]
    \le
    \exp\left(\frac{a^2\lambda^2\sigma^2}{2}\right).
\]
\end{proof}

In addition to the closure properties, we will need several additional standard results. The first is the fact that binomial random variables (with $p = 1/3$) are sub-Gaussian.
\begin{lemma}[Binomial sub-Gaussianity]
\label{lem:binomial-subg}
Let \(B\sim \mathrm{Bin}(n,1/3)\). Then \(B-\tfrac n3 \in \mathrm{SG}(\tfrac
n4)\). Equivalently, \(B\in \mathrm{SG}(\tfrac n3,\tfrac n4)\).
\end{lemma}

\begin{proof}
Write \(B=\sum_{i=1}^n \xi_i\), where the \(\xi_i\)'s are independent
\(\mathrm{Bernoulli}(1/3)\) random variables. Each centered summand
\(\xi_i-\mathbb{E}[\xi_i]\) is supported on an interval of length \(1\). By
straightforward calculation,
\[
    \xi_i-\mathbb{E}[\xi_i] \in \mathrm{SG}(1/4).
\]
The result follows from Lemma~\ref{lem:subg-sums}.
\end{proof}

Next, we consider what happens when we clip a sub-Gaussian random variable to a bounded range. The result, not surprisingly, is still sub-Gaussian with essentially the same variance; and, if the original mean was far from the clipping boundaries, then the mean is essentially preserved as well. 
\begin{restatable}[Clipping a sub-Gaussian]{lemma}{clippingsubg}
    \label{lem:clipping-subg}
    Let \(X\in \mathrm{SG}(\mu,v)\), and define
    \[
        \clip_B(X)=\max\{0,\min\{X,B\}\}.
    \]
    Then
    \[
        \clip_B(X)\in \mathrm{SG}\left(\mathbb E[\clip_B(X)],v\right).
    \]
    
    Moreover, for any constant \(C>0\), there exists a constant \(\gamma>0\),
    depending only on \(C\), such that if \(v\le CB\) and
    \[
        \mu\in \left[\frac15B,\frac45B\right],
    \]
    then
    \[
        \Pr[\clip_B(X)\ne X]\le e^{-\gamma B},
        \text{ and }
        \left|\mathbb E[\clip_B(X)]-\mu\right|
        \le B e^{-\gamma B}.
    \]
\end{restatable}

The proof of Lemma~\ref{lem:clipping-subg} is deferred to
Appendix~\ref{app:clipping-subg}.

Finally, we will also need the following two slightly more specialized results. The first says that, even though Poisson random variables are \emph{not} sub-Gaussian (because their tails are too fat), they can (in some parameter regimes) be coupled to a random variable that is. 

\begin{restatable}[Sub-Gaussian coupling for a Poisson random variable]{lemma}{poissonsubgcoupling}
\label{lem:poisson-subg-coupling}
There exist absolute constants \(C,c>0\) such that the following holds. Let
\(X\sim \mathrm{Pois}(\mu)\) with \(\mu\ge 1\). Then there exists a random
variable \(G\), coupled with \(X\), such that \(\mathbb{E}[G]=\mu\), \(G-\mu \in
\mathrm{SG}(C\mu)\), and \(\mathbb{P}(G=X)\ge 1-e^{-c\mu}\). Moreover, \(G\) may
be chosen to be integer-valued.
\end{restatable}

The proof of Lemma~\ref{lem:poisson-subg-coupling} is deferred to
Appendix~\ref{app:poisson-subg-coupling}.

The final lemma considers what happens if we take a non-negative integer-valued sub-Gaussian random variable $X$ and replace it with $\mathrm{Bin}(X,1/3)$ (this is analogous to taking the overflow from a bin $q_1$ and deciding which of the elements should be forwarded to a given $q_2$). Here, critically, the result is not just sub-Gaussian, but it is sub-Gaussian with a \emph{smaller} variance (assuming the original variance was large enough). 

\begin{restatable}[Sub-Gaussianity under binomial thinning]{lemma}{subgbinomialthinning}
    \label{lem:subg-binomial-thinning}
    Let \(G\) be a nonnegative integer-valued random variable such that \(G\in
    \mathrm{SG}(\mu,v)\) (in particular \(\mathbb{E}[G]=\mu\)), and suppose that
    \(0\le G\le B\) almost surely. Conditional on \(G\), let \(H\mid G \sim
    \mathrm{Bin}(G,1/3)\). Then
    \[
        H\in \mathrm{SG}\left(\frac{\mu}{3}, \frac{v}{9}+\frac{B}{4}\right).
    \]
\end{restatable}

The proof of Lemma~\ref{lem:subg-binomial-thinning} is deferred to
Appendix~\ref{app:subg-binomial-thinning}.

\subsection{Proof of the Static Placement Theorem}\label{sec:actual-proof}

Having reviewed the properties of sub-Gaussian random variables, we now turn to the proof of the Static Placement Theorem, restated below for convenience.

\staticplacement*

We first prove the theorem in a Poissonized model, where rather than assuming $n$ elements, we assume that the number of elements is a Poisson random variable 
\[
    N\sim \mathrm{Pois}(n)
\]
with mean $n$. These elements are then placed into the hash table using the same construction as in the fixed-$n$ model. 

For a bin \(q\), let
\[
    A_q
\]
denote the number of elements whose first hash is \(q\) -- we refer to $A_q$ as the \emph{primary-arrival variable} for bin $q$. Under Poissonization, the variables
\(\{A_q\}\) are independent Poisson random variables.

For \(q\in L_i\), let
\[
    F_q
\]
denote the number of elements forwarded into \(q\) from layer \(L_{i-1}\), with
\(F_q=0\) for \(q\in L_1\). For \(q\in L_i\), \(i<\ell\), let
\[
    V_q
\]
denote the number of elements overflowed out of \(q\) and forwarded to
\(L_{i+1}\).

Before continuing, it will be helpful to compute the expectations of the primary-arrival variables. For \(q\in L_1\),
\[
    \mathbb E[A_q]
    =
    \frac{np_1}{b_1}
    =
    \frac{s_1+\frac12s_2}{b_1}
    =
    \frac54B.
\]
For \(q\in L_i\), where \(2\le i<\ell\),
\[
    \mathbb E[A_q]
    =
    \frac{np_i}{b_i}
    =
    \frac{\frac12(s_i+s_{i+1})}{b_i}
    =
    \frac34B.
\]
Finally, for $q \in L_\ell$, we have $A_q = 0$ since no elements have first hashes in the final layer. 

\paragraph{Analyzing a proxy process.} We now define an auxiliary process that is easier to analyze. For each
non-final bin \(q\), we couple \(A_q\) to an integer-valued random variable
\(\widehat A_q\) using the Poisson Sub-Gaussian Coupling Lemma
(Lemma~\ref{lem:poisson-subg-coupling}). Thus, if
\[
    \mu_q=\mathbb E[A_q],
\]
then
\[
    \widehat A_q\in \mathrm{SG}(\mu_q,C_{\mathrm P}\mu_q)
\]
and
\[
    \Pr[\widehat A_q\ne A_q]\le e^{-c_{\mathrm P}\mu_q},
\]
where \(C_{\mathrm P},c_{\mathrm P}>0\) are absolute constants. We take these
couplings independently over all bins.

The proxy process is defined recursively. For bins \(q\in L_1\), set
\[
    \widehat F_q=0.
\]
For a non-final bin \(q\), define its proxy pre-overflow quantity by
\[
    \widehat Y_q=\widehat F_q+\widehat A_q-B,
\]
and define its proxy overflow by clipping:
\[
    \widehat V_q=\clip_B(\widehat Y_q)
    =
    \max\{0,\min\{\widehat Y_q,B\}\}.
\]
Conditioned on \(\widehat V_q\), the proxy process sends
\(\widehat V_q\) elements independently and uniformly to the three bins in the
next layer allowed by the definition of \(h_2\). For a bin \(r\in L_{i+1}\), let
\[
    \widehat H_{q,r}
\]
denote the number sent from \(q\) to \(r\). Thus, whenever \(r\) is one of the
three possible destinations of \(q\),
\[
    \widehat H_{q,r}\mid \widehat V_q
    \sim
    \mathrm{Bin}(\widehat V_q,1/3).
\]
Finally, for \(r\in L_{i+1}\), define
\[
    \widehat F_r=\sum_{q\in\mathsf{Pred}(r)} \widehat H_{q,r},
\]
where \(\mathsf{Pred}(r)\) is the set of predecessor bins in \(L_i\) that can forward
to \(r\). This represents the number of elements forwarded into $r$ from the previous layer in the proxy process.

Note that this proxy process differs from the true process in only two ways: (1) each $A_q$ is replaced with $\widehat A_q$, and (2) the number of balls $\widehat V_q$ that overflow out of a bin $q$ is artificially truncated to at most $B$. At the end of the analysis, we will show that, with high probability in $n$, neither of these distinctions ever matters (the two processes behave exactly the same). Until then, however, it will be more convenient to reason about the proxy process.

\paragraph{Local independence of predecessor bins.} We now argue that, within the proxy process, if we focus on a given bin $q \in L_{i}$ (for some $i \ge 2$), then the behaviors of $q$'s predecessor bins are mutually independent.

\begin{lemma}[Local independence of predecessor overflows]
    \label{lem:local-independence}
    Fix \(2\le i\le \ell\), and let \(q\in L_i\). Let
    $
        \mathsf{Pred}(q)\subseteq L_{i-1}
    $
    be the set of bins in layer \(L_{i-1}\) from which an element may be forwarded
    to \(q\). Then $|\mathsf{Pred}(q)|=6$ and the proxy overflow variables $\{\widehat V_{q'}:q'\in\mathsf{Pred}(q)\}$ are mutually independent. Finally, conditioned on the value of \(\widehat V_{q'}\), the number of proxy
    elements forwarded from \(q'\in\mathsf{Pred}(q)\) to \(q\) is distributed as
    $\mathrm{Bin}(\widehat V_{q'},1/3),
    $
    independently over \(q'\in\mathsf{Pred}(q)\).
    \end{lemma}
    
    \begin{proof}
    First consider the size of \(\mathsf{Pred}(q)\). A bin \(q'\in L_{i-1}\) can forward
    to \(q\in L_i\) only if deleting the final bit of the address of \(q'\), and
    then changing the bit-pair in positions
    \[
        2i-3,\ 2i-2,
    \]
    produces the address of \(q\). There are three possible original values for
    this bit-pair, since it must be different from the pair in \(q\), and there are
    two possible values for the deleted bit. Hence $|\mathsf{Pred}(q)|=3\cdot 2=6$.
    
    We now prove the independence claim. For a bin \(u\), let
    \[
        \mathcal D(u)
    \]
    denote the set of bins \(r\) whose primary-arrival variables \(\widehat A_r\)
    can influence the proxy overflow \(\widehat V_u\) (so $\mathcal D(u)$ consists of the predecessors of $u$, the predecessors of those predecessors, etc.). We claim that, for the six
    bins \(q'\in\mathsf{Pred}(q)\), the sets
    \[
        \{\mathcal D(q') : q'\in\mathsf{Pred}(q)\}
    \]
    are pairwise disjoint. Since the proxy primary-arrival variables
    \(\{\widehat A_r\}_r\) are mutually independent, this implies that the overflows
    \[
        \{\widehat V_{q'}:q'\in\mathsf{Pred}(q)\}
    \]
    are mutually independent.
    
    We prove the disjointness claim by describing \(\mathcal D(u)\). Let
    \(u\in L_{i-1}\). The overflow \(\widehat V_u\) is determined by the proxy
    primary arrivals in the bins that can reach \(u\) through the first \(i-2\)
    forwarding steps. Equivalently, \(\mathcal D(u)\) consists of those bins in
    layers
    \[
        L_1,L_2,\ldots,L_{i-1}
    \]
    whose address can be transformed into the address of \(u\) by repeatedly
    applying the earlier layer transitions.
    
    Crucially, the earlier transitions use only the bit-pairs
    \[
        (1,2),(3,4),\ldots,(2i-5,2i-4),
    \]
    and delete only one trailing bit at each step. They do not inspect or modify the
    fresh bit-pair
    \[
        (2i-3,2i-2)
    \]
    used in the transition from \(L_{i-1}\) to \(L_i\).
    
    The six bins in \(\mathsf{Pred}(q)\) differ from one another either in this fresh
    bit-pair \((2i-3,2i-2)\), or in the final bit that is deleted in the transition
    to \(L_i\). Since neither of these bits is touched by any earlier transition,
    any bin that can reach one predecessor \(q'\in\mathsf{Pred}(q)\) through earlier
    transitions cannot also reach a different predecessor \(q''\in\mathsf{Pred}(q)\).
    Thus
    \[
        \mathcal D(q')\cap \mathcal D(q'')=\emptyset
        \qquad
        \text{for distinct }q',q''\in\mathsf{Pred}(q).
    \]
    This proves the disjointness claim, and hence the mutual independence of the
    proxy overflows \(\widehat V_{q'}\).
    
    Finally, in the proxy process,
    conditioned on \(\widehat V_{q'}\), each of the \(\widehat V_{q'}\) overflowed
    proxy elements independently chooses one of the three allowed destinations in
    the next layer uniformly at random. Therefore the number sent from \(q'\) to
    the fixed bin \(q\) is
    \[
        \mathrm{Bin}(\widehat V_{q'},1/3),
    \]
    and this is independent over the six bins $q'\in\mathsf{Pred}(q)$.
    \end{proof}

\paragraph{Analyzing overflow variances. } We now prove the most important result of the analysis: that the overflow random variables $\widehat V_q$ are sub-Gaussian with variance $O(B)$. 

\begin{lemma}[Variance invariant]
\label{lem:proxy-variance-invariant}
There is an absolute constant \(C>0\) such that the following holds for the
proxy process.

For every non-final bin \(q\),
\[
    \widehat V_q
    \in
    \mathrm{SG}\left(\mathbb E[\widehat V_q],CB\right).
\]
Moreover, for every bin \(q\in L_i\), \(i\ge 2\),
\[
    \widehat F_q
    \in
    \mathrm{SG}\left(\mathbb E[\widehat F_q],
        \left(\frac23C+\frac32\right)B
    \right).
\]
\end{lemma}

\begin{proof}
Let \(C_{\mathrm P}\) be the absolute constant from the Poisson Sub-Gaussian
Coupling Lemma (Lemma~\ref{lem:poisson-subg-coupling}). Choose \(C\) large
enough that
\[
    \frac54C_{\mathrm P}\le C
\]
and
\[
    \frac23C+\frac32+\frac34C_{\mathrm P}\le C.
\]
This is possible because the coefficient of \(C\) in the second inequality is
\(2/3<1\).

We prove the statement by induction over the layers. We begin with the base case of $q \in L_1$. For \(q\in L_1\), we have \(\widehat F_q=0\), and
\[
    \widehat A_q
    \in
    \mathrm{SG}\left(\frac54B,\frac54C_{\mathrm P}B\right)
    \subseteq
    \mathrm{SG}\left(\frac54B,CB\right).
\]
Thus
\[
    \widehat Y_q=\widehat A_q-B
\]
has sub-Gaussian variance at most \(CB\). Since
\[
    \widehat V_q=\clip_B(\widehat Y_q),
\]
the Clipping Lemma
(Lemma~\ref{lem:clipping-subg}) gives
\[
    \widehat V_q
    \in
    \mathrm{SG}\left(\mathbb E[\widehat V_q],CB\right).
\]

Now consider a bin \(q\in L_i\), where \(i\ge 2\), and assume the invariant has
been proved for layer \(L_{i-1}\). By the Local Independence of Predecessor
Bins Lemma (Lemma~\ref{lem:local-independence}), \(q\) has six independent
predecessors. For each \(q'\in\mathsf{Pred}(q)\), conditioned on \(\widehat V_{q'}\),
\[
    \widehat H_{q',q}\mid \widehat V_{q'}
    \sim
    \mathrm{Bin}(\widehat V_{q'},1/3).
\]
Since
\[
    0\le \widehat V_{q'}\le B
\]
and
\[
    \widehat V_{q'}
    \in
    \mathrm{SG}\left(\mathbb E[\widehat V_{q'}],CB\right),
\]
the Sub-Gaussianity under Binomial Thinning Lemma
(Lemma~\ref{lem:subg-binomial-thinning}) gives
\[
    \widehat H_{q',q}
    \in
    \mathrm{SG}\left(
        \frac{\mathbb E[\widehat V_{q'}]}{3},
        \frac{CB}{9}+\frac{B}{4}
    \right).
\]
Summing over the six independent predecessors and using the Sums of Independent
Sub-Gaussians Lemma (Lemma~\ref{lem:subg-sums}), we obtain
\[
    \widehat F_q
    \in
    \mathrm{SG}\left(
        \mathbb E[\widehat F_q],
        6\left(\frac{CB}{9}+\frac{B}{4}\right)
    \right).
\]
Equivalently,
\[
    \widehat F_q
    \in
    \mathrm{SG}\left(
        \mathbb E[\widehat F_q],
        \left(\frac23C+\frac32\right)B
    \right).
\]

This gives us the desired result for $\widehat F_q$, and it remains to consider $\widehat V_q$ (in the case where $i < \ell$). Recall that 
\[
    \widehat A_q
    \in
    \mathrm{SG}\left(\frac34B,\frac34C_{\mathrm P}B\right).
\]
Therefore, by the Sums of Independent Sub-Gaussians Lemma (Lemma~\ref{lem:subg-sums}),
\[
    \widehat Y_q
    =
    \widehat F_q+\widehat A_q-B
\]
has variance at most
\[
    \left(\frac23C+\frac32+\frac34C_{\mathrm P}\right)B
    \le
    CB.
\]
Applying the Clipping Lemma
(Lemma~\ref{lem:clipping-subg}) to
\[
    \widehat V_q=\clip_B(\widehat Y_q)
\]
gives
\[
    \widehat V_q
    \in
    \mathrm{SG}\left(\mathbb E[\widehat V_q],CB\right),
\]
which completes the proof. 
\end{proof}

\paragraph{Analyzing overflow means. } We next compute the means in the proxy process. By symmetry, all bins in the
same layer have the same proxy-overflow mean. Let
\[
    m_i=\mathbb E[\widehat V_q]
    \qquad
    \text{for } q\in L_i,\ i<\ell.
\]

\begin{lemma}[Proxy mean invariant]
\label{lem:proxy-mean-invariant}
There is a constant \(\gamma>0\) such that the following holds. Define
\[
    \tau=Be^{-\gamma B},
\]
and define
\[
    \Delta_1=\tau,
    \qquad
    \Delta_i=2\Delta_{i-1}+\tau
    \quad\text{for } i\ge 2.
\]
Assume that
\[
    \Delta_i\le \frac{B}{40}
    \qquad
    \text{for all } i<\ell.
\]
Then, for every non-final layer \(L_i\),
\[
    m_i=\frac14B+\varepsilon_i
\]
for some deterministic \(\varepsilon_i\) satisfying
\[
    |\varepsilon_i|\le \Delta_i.
\]
\end{lemma}

\begin{proof}
Choose \(\gamma>0\) small enough so that the Clipping Lemma
(Lemma~\ref{lem:clipping-subg}) applies with variance \(CB\) and
gives mean shift at most
\[
    \tau=Be^{-\gamma B}.
\]
We prove the claim by induction.

For \(q\in L_1\),
\[
    \mathbb E[\widehat Y_q]
    =
    \mathbb E[\widehat A_q]-B
    =
    \frac54B-B
    =
    \frac14B.
\]
Moreover, by Lemma~\ref{lem:proxy-variance-invariant},
\[
    \widehat Y_q
\]
has variance at most \(CB\). The mean \(\frac14B\) lies far from both
clipping thresholds \(0\) and \(B\), so the Clipping Lemma
(Lemma~\ref{lem:clipping-subg}) gives
\[
    \left|
        \mathbb E[\widehat V_q]-\frac14B
    \right|
    \le \tau.
\]
Thus \(|\varepsilon_1|\le \Delta_1\).

Now suppose \(2\le i<\ell\), and assume
\[
    m_{i-1}=\frac14B+\varepsilon_{i-1},
    \qquad
    |\varepsilon_{i-1}|\le \Delta_{i-1}.
\]
For \(q\in L_i\), the expected forwarded load is
\[
    \mathbb E[\widehat F_q]
    =
    6\cdot \frac13 m_{i-1}
    =
    2m_{i-1}.
\]
The expected primary-arrival proxy $\E[\widehat A_q]$ is $\frac34B$, so
\[
    \mathbb E[\widehat Y_q]
    =
    2m_{i-1}+\frac34B-B
    =
    \frac14B+2\varepsilon_{i-1}.
\]
Since
\[
    |2\varepsilon_{i-1}|
    \le
    2\Delta_{i-1}
    \le
    \frac{B}{20},
\]
this mean lies in \([B/5,4B/5]\). By the Variance Invariant
(Lemma~\ref{lem:proxy-variance-invariant}), \(\widehat Y_q\) has variance
at most \(CB\). Hence the Clipping Lemma
(Lemma~\ref{lem:clipping-subg}) gives
\[
    \left|
        \mathbb E[\widehat V_q]
        -
        \left(\frac14B+2\varepsilon_{i-1}\right)
    \right|
    \le \tau.
\]
Writing
\[
    \mathbb E[\widehat V_q]=\frac14B+\varepsilon_i,
\]
we get
\[
    |\varepsilon_i-2\varepsilon_{i-1}|\le \tau,
\]
and therefore
\[
    |\varepsilon_i|
    \le
    2|\varepsilon_{i-1}|+\tau
    \le
    2\Delta_{i-1}+\tau
    =
    \Delta_i.
\]
This completes the induction.
\end{proof}

Recall that the bin size $B$ is given by $B = \lceil c \log n \rceil$ for some large positive constant $c$. Moreover, the constants used in the previous arguments ($C, \gamma, C_{\mathrm P}, c_{\mathrm P}$) are all independent of $c$. Therefore, we may treat those constants as fixed, and take $c$ to be sufficiently large. Doing so, we can argue that the error sequence
in Lemma~\ref{lem:proxy-mean-invariant} is negligible. Since
\[
    \Delta_i=(2^i-1)\tau
\]
and
\[
    2^\ell=\Theta(n^{1/4}),
\]
we have
\[
    \max_{i<\ell}\Delta_i
    \le
    O\left(Bn^{1/4}e^{-\gamma B}\right).
\]
Because \(B=c\log n\), choosing \(c\) sufficiently large ensures that
\[
    \max_{i<\ell}\Delta_i\le Bn^{-10}.
\]
In particular, the hypothesis
\[
    \Delta_i\le B/40
\]
from Lemma~\ref{lem:proxy-mean-invariant} holds for every \(i<\ell\).

\paragraph{Proving the theorem in the Poissonized model. } We now have all the pieces necessary to prove the theorem in the Poissonized model. 

\begin{lemma}[Poissonized static placement]
    \label{lem:poissonized-static-placement}
    In the Poissonized model, the conclusions of the Static Placement Theorem hold
    with probability at least \(1-n^{-D}\), for any desired constant \(D>0\),
    provided the constant $c$ in \(B=c\log n\) is chosen sufficiently large.
    \end{lemma}
    
    \begin{proof}
    Let \(\mathcal G\) be the event that all of the following hold simultaneously.
    
    \begin{enumerate}
        \item For every non-final bin \(q\),
        \[
            \widehat A_q=A_q.
        \]
    
        \item For every bin \(q\in L_i\), \(2\le i\le \ell\),
        \[
            \frac14B\le \widehat F_q\le \frac34B.
        \]
    
        \item For every non-final bin \(q\),
        \[
            0\le \widehat Y_q\le B.
        \]
    \end{enumerate}
    Note that, should $\mathcal G$ occur, then the proxy process and the true process behave exactly the same. Indeed, the only differences between the two processes are that (1) each $A_q$ is replaced with $\widehat A_q$, and (2) the number of balls $\widehat V_q$ that overflow out of a bin $q$ is artificially truncated to at most $B$. If $\mathcal G$ occurs, then these differences never matter. 
    
    Now, let us reason about the probability of $\mathcal G$ failing. There are \(O(n/B)\) bins in total. The first condition fails for a fixed bin
    with probability \(e^{-\Omega(B)}\), by the Poisson Sub-Gaussian Coupling Lemma
    (Lemma~\ref{lem:poisson-subg-coupling}).
    
    For the second condition, the variance and mean invariants imply
    \[
        \mathbb E[\widehat F_q]
        =
        \frac12B+o(B),
    \]
    and
    \[
        \widehat F_q-\mathbb E[\widehat F_q]
        \in \mathrm{SG}(O(B)).
    \]
    Thus
    \[
        \Pr\left[
            \widehat F_q\notin\left[\frac14B,\frac34B\right]
        \right]
        \le e^{-\Omega(B)}.
    \]
    
    For the third condition, the variance and mean invariants imply
    \[
        \mathbb E[\widehat Y_q]
        =
        \frac14B+o(B),
    \]
    and
    \[
        \widehat Y_q-\mathbb E[\widehat Y_q]
        \in \mathrm{SG}(O(B)).
    \]
    Thus
    \[
        \Pr[\widehat Y_q<0]+\Pr[\widehat Y_q>B]
        \le e^{-\Omega(B)}.
    \]
    
    By a union bound over all \(O(n/B)\) bins, the event \(\mathcal G\) holds with
    probability at least
    \[
        1-O(n/B)e^{-\Omega(B)}.
    \]
    Choosing \(c\) sufficiently large makes this at least \(1-n^{-D}\).
    
    It remains to check that, on \(\mathcal G\), the proxy process (and therefore the true process) satisfies the theorem's conclusions. Because the two processes are identical on $\mathcal G$, we will use the proxy-process variables and the true-process variables interchangeably. On \(\mathcal G\), every forwarded element is successfully placed,
    because every bin receives at most
    \[
        \frac34B<B
    \]
    forwarded elements. Also, every non-final bin fills during its own phase,
    because
    \[
        0\le \widehat Y_q
    \]
    implies
    \[
        F_q+A_q\ge B.
    \]
    Finally, from the second condition of $\mathcal G$, we have
    \[
        \frac14B\le F_q\le \frac34B
    \]
    for every bin after the first layer, which is the final property required by the theorem.
    \end{proof}

\paragraph{De-Poissonizing. }
It remains to de-Poissonize. Let \(\mathcal E\) be the event that the theorem's
three conclusions hold. In the Poissonized model, conditioned on \(N=n\), the
hash choices have exactly the same distribution as in the original fixed-\(n\)
model. Hence
\[
    \mathbb P_n[\mathcal E^c]
    =
    \mathbb P_{\mathrm{Pois}}[\mathcal E^c\mid N=n].
\]
Therefore
\[
    \mathbb P_n[\mathcal E^c]
    \le
    \frac{\mathbb P_{\mathrm{Pois}}[\mathcal E^c]}{\mathbb P[N=n]}.
\]
By Stirling's formula,
\[
    \mathbb P[N=n]
    =
    e^{-n}\frac{n^n}{n!}
    =
    \Theta(n^{-1/2}).
\]
Thus, if the Poissonized failure probability is at most \(n^{-D-1}\), then the
fixed-\(n\) failure probability is at most \(n^{-D-1/2}\). Since \(D\) was
arbitrary, the theorem follows with failure probability \(1/\poly(n)\).

\subsection{Analyzing \(n' < n\) Keys}\label{sec:underloaded}

We now consider what happens when the table, constructed for capacity parameter
\(n\), is used to store only \(n'\le n\) keys. The hash functions and layer
probabilities are still those defined using the original parameter \(n\). Thus,
compared to the preceding section, the table layout is unchanged, but fewer
keys are inserted. Our goal is to prove Theorem~\ref{thm:static-placement-underloaded},
restated below for convenience.

\staticplacementunderloaded*

\begin{proof}
We prove the theorem for a fixed value of \(n'\le n\). We begin by arguing that the first two conclusions are monotonic consequences of the Static Placement Theorem.

Fix the oblivious tie-breaking rule used inside each bin: among elements with
the same first-choice bin, the rule deciding which elements stay and which are
forwarded is independent of their second hashes. Under any such fixed rule, the
placement process is monotone under deletion of keys. More precisely, if we run
the process on a set \(S\) of keys and on a subset \(S'\subseteq S\), then, by
induction over the phases, every bin receives no more forwarded elements in the
run on \(S'\) than in the run on \(S\). Indeed, in a given bin, deleting keys
can only decrease the number of elements already present and can only decrease
the number of primary arrivals. Therefore the set of primary arrivals forwarded
out of that bin is a subset of the set that would have been forwarded in the
larger run.

Now couple the \(n'\)-key instance with an \(n\)-key instance by adding
\(n-n'\) independent extra keys. Whenever the \(n\)-key instance satisfies the
Static Placement Theorem, the \(n'\)-key instance also has no failed forwarded
placement, and every $F_q$ for every bin $q$ in layers $L_i$, $i > 1$, is at most $\frac{3B}{4}$. Moreover, since every bin in the final layer has at most as many
elements in the \(n'\)-key run as in the \(n\)-key run, and every final-layer
bin has at most \(3B/4\) forwarded elements in the successful \(n\)-key run, all
final-layer bins have free slots in the \(n'\)-key run. Thus the first two
claims hold with probability \(1-1/\poly(n)\).

It remains to prove the third claim. We use the same Poissonization and proxy
process as in the proof of the Static Placement Theorem. Namely, we first
consider a Poissonized model with
\[
    N'\sim \mathrm{Pois}(n')
\]
keys, and then de-Poissonize at the end. The proxy process is defined exactly
as before, except that the primary-arrival variables now have smaller means:
if
\[
    \alpha=\frac{n'}{n}\le 1,
\]
then for \(q\in L_1\),
\[
    \mathbb E[A_q]=\alpha\cdot \frac54B,
\]
and for \(q\in L_i\), \(2\le i<\ell\),
\[
    \mathbb E[A_q]=\alpha\cdot \frac34B\le \frac34B.
\]
The variance argument from Lemma~\ref{lem:proxy-variance-invariant} continues
to apply without change, since it only uses the fact that the primary-arrival
variables have sub-Gaussian variance \(O(B)\), and these proxies only get
smaller when \(\alpha\le 1\).

In particular, there is a constant \(C_0>0\) such that, in the Poissonized proxy
process, for every directed predecessor edge \(q'\to q\),
\[
    \widehat H_{q',q}
    \in
    \mathrm{SG}\left(
        \mathbb E[\widehat H_{q',q}],
        C_0B
    \right),
\]
and for every non-final bin \(q\),
\[
    \widehat A_q
    \in
    \mathrm{SG}\left(
        \mathbb E[\widehat A_q],
        C_0B
    \right).
\]
Moreover, the proxy process agrees with the true process with probability
\(1-1/\poly(n)\), by the same comparison argument as in
Lemma~\ref{lem:poissonized-static-placement}: the probabilities that
\(\widehat A_q\ne A_q\), that an upper truncation is active, or that a forwarded
load exceeds \(B\), are all \(e^{-\Omega(B)}\) per bin, and there are only
\(O(n/B)\) bins.

Let
\[
    \xi=\frac1{100}.
\]
Since \(B=c\log n\), by taking \(c\) sufficiently large and applying a union
bound over all bins and all predecessor edges, we may assume that the following
events hold simultaneously with probability \(1-1/\poly(n)\):

\begin{enumerate}
    \item The proxy process and the true process agree.

    \item For every directed predecessor edge \(q'\to q\),
    \[
        \left|
            H_{q',q}
            -
            \mathbb E[\widehat H_{q',q}]
        \right|
        \le \xi B.
    \]

    \item For every bin \(q\in L_i\), \(2\le i<\ell\),
    \[
        A_q\le \left(\frac34+\xi\right)B.
    \]
\end{enumerate}
For final-layer bins, \(A_q=0\), so the third condition is not needed there.

We now show that these events imply the desired structural property. Fix
\(q\in L_i\) with \(2\le i\le \ell\). If \(q\) has a free slot, there is nothing
to prove. So suppose \(q\) is full.

By the second conclusion already proved, every final-layer bin has a free slot.
Hence \(q\) cannot lie in \(L_\ell\); we must have \(i<\ell\). Since \(q\) is
full,
\[
    A_q+\sum_{q'\in\mathsf{Pred}(q)} H_{q',q}\ge B.
\]
By the third good event,
\[
    A_q\le \left(\frac34+\xi\right)B.
\]
Therefore
\[
    \sum_{q'\in\mathsf{Pred}(q)} H_{q',q}
    \ge
    \left(\frac14-\xi\right)B.
\]

By symmetry, the six random variables
\[
    \widehat H_{q',q},
    \qquad q'\in\mathsf{Pred}(q),
\]
have the same expectation. Let this common expectation be \(m_q\). The second
good event gives
\[
    H_{q',q}\le m_q+\xi B
    \qquad
    \text{for every }q'\in\mathsf{Pred}(q).
\]
Since \(|\mathsf{Pred}(q)|=6\),
\[
    6(m_q+\xi B)
    \ge
    \sum_{q'\in\mathsf{Pred}(q)} H_{q',q}
    \ge
    \left(\frac14-\xi\right)B.
\]
Thus
\[
    m_q
    \ge
    \left(\frac{1}{24}-\frac{7}{6}\xi\right)B.
\]
Using \(\xi=1/100\), this gives
\[
    m_q\ge \frac{3}{100}B.
\]
Applying the second good event again,
\[
    H_{q',q}\ge m_q-\xi B\ge \frac{2}{100}B
    \qquad
    \text{for every }q'\in\mathsf{Pred}(q).
\]
Thus the theorem holds with, for example,
\[
    \eta=\frac1{50}.
\]

This proves the theorem in the Poissonized model with failure probability
\(1/\poly(n)\). To pass to the fixed-\(n'\) model, condition on the event
\(N'=n'\). For \(n'\ge 1\),
\[
    \Pr[N'=n']=\Theta\left((n')^{-1/2}\right)\ge n^{-1/2}.
\]
Thus a polynomially small failure probability in the Poissonized model remains
polynomially small after conditioning. This
completes the proof.
\end{proof}

\section{Conclusion}
The price of cuckoo hashing's $O(\log\epsilon^{-1})$ worst-case queries has always been slow insertions: the best known bound was $O(\epsilon^{-1})$ expected time~\cite{KM25}, which was conjectured to be optimal up to polylogarithmic factors~\cite{bender2024tight}. Braided cuckoo hashing removes this cost, supporting insertions and deletions in worst-case time $O(\log\epsilon^{-1})$ while preserving the worst-case query time. Perhaps the largest remaining impediment to realizing these gains in practice is the constant factor on the number of probes a query makes: $d$-ary cuckoo hashing~\cite{FP12,xorsatthresholds} supports queries using $(1+o(1))\ln\epsilon^{-1}$ probes, which is optimal even among static hash tables by a simple counting argument. This leaves the following important open question: is an insertion/deletion time of $O(\log\epsilon^{-1})$ possible such that all queries use $(1+o(1))\ln\epsilon^{-1}$ probes? 

The following additional questions also warrant investigation. Can braided cuckoo hashing be extended to support (i) dynamic resizing? (ii) an \emph{expected} query time of $O(1)$? (iii) a load factor of exactly $1$? (iv) dynamically adapting insertion, deletion, and/or query times as $\epsilon$ varies? We expect that (i) can be implemented using the same basic approach as in \cite{bender2024tight}, but there are many details one would need to work out, so we have deferred it to future work. The construction of \cite{bender2024tight} also provides a potential approach to (ii): by using their construction in each bin, it would seem the expected query time could be brought to $O(1)$. However, this would increase the expected insertion time by a factor $O(\log\log B) = O(\log\log\log\epsilon^{-1})$ and the worst-case insertion time by a potentially much larger factor. (iii) and (iv) would seem to require new techniques.

\section*{Acknowledgments}

\ifanonymous\else
This work was partially supported by the Jeanne~B.\ and Richard~F.\ Berdik ARCS Pittsburgh Endowed Scholar Award, NSF grants CNS-2504471 and CCF-2542165, a Jane Street Research Grant, the David and Lucile Packard Foundation (Award \#202071730), and gifts from Nikolai Mushegian and Microsoft.

\fi
Several of the technical sections in this paper were written with the help of GPT 5.5 Pro, which was used to turn detailed section outlines into initial drafts. These initial drafts were then heavily edited by the authors of the paper. Claude Opus 5.5 was used to check the paper for grammatical errors and to polish the writing.

\bibliographystyle{alpha}
\bibliography{references}

\appendix

\section{Proof of Theorem~\ref{thm:general-capacities}}
\label{app:general-capacities}

In this section, we present the proof of Theorem \ref{thm:general-capacities}.

\generalcaps*

\begin{proof}

Let $c$ be a sufficiently large positive constant, chosen so that $c \log n$ is a power of two, and define $B = c \log n$.
The data structure will consist of $\Theta(n^{1/5})$ sub-structures, each with capacity $N = \Theta(n^{4/5})$, and each using bin size $B = \Theta(\log n) = \Theta(\log N)$. Each sub-structure will be implemented, as a black box, using Theorem \ref{thm:dynamic-whi} (or, if one wants strong history independence, Theorem \ref{thm:dynamic-shi}). 
We refer to the sub-structures as \defn{blocks}. Each element will hash to a random block and interact only with that block. So long as the blocks are weakly history-independent (resp. strongly history-independent),  the data structure will also be weakly history-independent (resp. strongly history-independent).

The main thing to be careful about is to select the block size $N$ to satisfy the Divisibility Assumption. Set
\[
    \ell_0=\left\lfloor \frac15\log_2 n\right\rfloor,
    \qquad
    K_0=\left\lfloor \frac35\log_2 n-\log_2 B\right\rfloor,
\]
and define
\[
    N=B2^{K_0}\left(2^{\ell_0}-\frac32\right) = \Theta(n^{4/5}).
\]
By design, the Divisibility Assumption is satisfied, since
\[
    \ell_0=\frac14\log_2 N+O(1),
    \qquad
    K_0=\frac34\log_2 N- \log_2 B \pm O(1).
\]
Each block has capacity $N$ and has $M = N + \Theta(N^{3/4}) = N + \Theta(n^{3/5})$ slots. 
Moreover, so long as we can guarantee that the block never has more than 
$N$ live elements at a time, then Theorem \ref{thm:dynamic-whi} (or Theorem \ref{thm:dynamic-shi})
guarantees $O(\log n)$-time operations with high probability in $N$ (and therefore also with high probability in $n$).

Next we determine \emph{how many} blocks to use. We want to have enough blocks that, with high probability, 
each block stays below its capacity $N$ at all times. To do this, set
\[
    \Delta=\sqrt{N}\log n,
    \qquad
    r=\left\lceil \frac{n}{N-\Delta}\right\rceil.
\]
The global table consists of
\(r\) blocks.  The number
of slots used is \(m_0=rM\).  Since \(\Delta=o(N)\),
\[
    rN\le \left(\frac{n}{N-\Delta}+1\right)N
       = n+O\left(\frac{n\Delta}{N}+N\right)
       = n+O\left(\frac{n \cdot (\sqrt{n^{4/5}} \log n)}{n^{4/5}}+n^{4/5}\right) = n+O(n^{4/5}).
\]
The total number of slots $m_0 = rM$ therefore satisfies
$$m_0 \le n + O(n^{4/5}) + r \cdot (M - N) = n + O(n^{4/5}) + O(n^{1/5} \cdot n^{3/5}) = n + O(n^{4/5}).$$

To prove the theorem, it remains to check that, during a sequence of $T = n^{O(1)}$ insertions and deletions,
there is (with high probability in $n$) never a point in time at which some block has more than $N$ live keys
at a time. Since the sequence is only of polynomial length, and since there are $O(n/ B)$ blocks,
 it suffices to consider the state of the hash table at a fixed moment in time, to consider a fixed block $j$,
 and to show that, with high probability in $n$, there are at most $N$ keys that hash to $j$.

Let
$
    \ell:U\to\{1,\ldots,r\}
$
be the hash function used to map elements to blocks. For a fixed block
\(j\), the number $K$ of live keys \(x\) with \(\ell(x)=j\) is a sum of i.i.d.\ Bernoulli
random variables with total mean at most $n/r \le N - \Delta = N - \sqrt{N} \log n$. By a Chernoff bound,
$$\Pr[K \ge N] \le \Pr[K \ge \E[K] + \sqrt{N} \log n] \le e^{-\Omega(\log^2 n)} \le 1 / \poly(n).$$
Therefore, with high probability in $n$, there are at most $N$ keys that hash to the block at a time. This completes the proof.
\end{proof}

\section{Proof of Lemma~\ref{lem:subg-tail-mgf}}
\label{app:subg-tail-mgf}

\subgtailmgf*

\begin{proof}
Let \(X:=Y-\mu\). Then \(\mathbb{E}[X]=0\).

For the first part, assume \(X\in \mathrm{SG}(\sigma^2)\). For every
\(\lambda>0\), Markov's inequality gives
\[
    \mathbb{P}[X\ge t]
    =
    \mathbb{P}[e^{\lambda X}\ge e^{\lambda t}]
    \le
    e^{-\lambda t}\mathbb{E}[e^{\lambda X}]
    \le
    \exp\left(-\lambda t+\frac{\lambda^2\sigma^2}{2}\right).
\]
Optimizing over \(\lambda>0\), with \(\lambda=t/\sigma^2\), gives
\[
    \mathbb{P}[X\ge t]
    \le
    \exp\left(-\frac{t^2}{2\sigma^2}\right).
\]
Applying the same argument to \(-X\), which is also in
\(\mathrm{SG}(\sigma^2)\), gives
\[
    \mathbb{P}[X\le -t]
    \le
    \exp\left(-\frac{t^2}{2\sigma^2}\right).
\]
Therefore, by the union bound,
\[
    \mathbb{P}[|X|\ge t]
    \le
    2\exp\left(-\frac{t^2}{2\sigma^2}\right).
\]

For the converse, assume that
\[
    \mathbb{P}[|X|\ge t]
    \le
    2\exp\left(-\frac{t^2}{v}\right)
    \qquad \text{for all } t\ge 0.
\]
We first derive a moment bound. For every integer \(k\ge 1\),
\[
\begin{aligned}
    \mathbb{E}[|X|^k]
    &=
    \int_0^\infty
        \mathbb{P}[|X|^k\ge s]\,ds \\
    &=
    k\int_0^\infty
        t^{k-1}\mathbb{P}[|X|\ge t]\,dt \\
    &\le
    2k\int_0^\infty
        t^{k-1}\exp\left(-\frac{t^2}{v}\right)\,dt.
\end{aligned}
\]
Changing variables \(u=t^2/v\), we get
\[
    \mathbb{E}[|X|^k]
    \le
    k v^{k/2}\Gamma(k/2).
\]
Using the standard estimate \(\Gamma(k/2)\le C_0^k k^{k/2}\)
for an absolute constant \(C_0>0\), this implies
\[
    \mathbb{E}[|X|^k]
    \le
    (C_1\sqrt{vk})^k
\]
for an absolute constant \(C_1>0\).

Now fix \(\lambda\in\mathbb{R}\). Since \(\mathbb{E}[X]=0\),
\[
\begin{aligned}
    \mathbb{E}[e^{\lambda X}]
    &=
    1+\sum_{k=2}^\infty
        \frac{\lambda^k\mathbb{E}[X^k]}{k!} \\
    &\le
    1+\sum_{k=2}^\infty
        \frac{|\lambda|^k\mathbb{E}[|X|^k]}{k!} \\
    &\le
    1+\sum_{k=2}^\infty
        \frac{(C_1|\lambda|\sqrt{vk})^k}{k!}.
\end{aligned}
\]
Using Stirling's lower bound \(k!\ge (k/e)^k\), we have
\[
    \frac{(C_1|\lambda|\sqrt{vk})^k}{k!}
    \le
    \left(\frac{C_2|\lambda|\sqrt v}{\sqrt k}\right)^k
\]
for an absolute constant \(C_2>0\). A standard summation estimate gives
\[
    \sum_{k=2}^\infty
        \left(\frac{C_2|\lambda|\sqrt v}{\sqrt k}\right)^k
    \le
    \exp(C_3\lambda^2 v)-1
\]
for an absolute constant \(C_3>0\). Therefore
\[
    \mathbb{E}[e^{\lambda X}]
    \le
    \exp(C_3\lambda^2 v).
\]
Equivalently,
\[
    \mathbb{E}[e^{\lambda(Y-\mu)}]
    \le
    \exp\left(\frac{\lambda^2(Cv)}{2}\right)
\]
for \(C=2C_3\). Hence \(Y\in \mathrm{SG}(\mu,Cv)\), as claimed.
\end{proof}

\section{Proof of Lemma~\ref{lem:clipping-subg}}
\label{app:clipping-subg}

\clippingsubg*

\begin{proof}
Let
\[
    \phi(x)=\max\{0,\min\{x,B\}\}.
\]
The map \(\phi\) is nondecreasing and \(1\)-Lipschitz. We use the standard
one-dimensional contraction principle: if \(\phi\) is nondecreasing and
\(1\)-Lipschitz, then
\[
    \phi(X)-\mathbb E[\phi(X)]
\]
is dominated by
\[
    X-\mathbb E[X]
\]
in convex order. Therefore, for every convex function \(\Psi\),
\[
    \mathbb E\left[
        \Psi\left(\phi(X)-\mathbb E[\phi(X)]\right)
    \right]
    \le
    \mathbb E\left[
        \Psi\left(X-\mathbb E[X]\right)
    \right].
\]
Applying this with
\[
    \Psi(z)=e^{\lambda z},
\]
for arbitrary \(\lambda\in\mathbb R\), gives
\[
    \mathbb E\left[
        e^{\lambda(\clip_B(X)-\mathbb E[\clip_B(X)])}
    \right]
    \le
    \mathbb E\left[
        e^{\lambda(X-\mu)}
    \right]
    \le
    \exp\left(\frac{\lambda^2v}{2}\right).
\]
Hence
\[
    \clip_B(X)\in \mathrm{SG}\left(\mathbb E[\clip_B(X)],v\right).
\]

We now prove the two additional bounds. Assume \(v\le CB\) and
\[
    \mu\in\left[\frac15B,\frac45B\right].
\]
Then both clipping thresholds \(0\) and \(B\) are at distance at least \(B/5\)
from \(\mu\). By the sub-Gaussian tail bound,
\[
    \Pr[\clip_B(X)\ne X]
    =
    \Pr[X<0]+\Pr[X>B]
    \le
    4\exp\left(-\frac{(B/5)^2}{2CB}\right)
    =
    4\exp\left(-\frac{B}{50C}\right).
\]

It remains to bound the change in expectation. Since clipping changes \(X\)
only when \(X<0\) or \(X>B\),
\[
    \left|\mathbb E[\clip_B(X)]-\mu\right|
    =
    \left|\mathbb E[\clip_B(X)-X]\right|
    \le
    \mathbb E[(-X)\mathbf 1_{X<0}]
    +
    \mathbb E[(X-B)\mathbf 1_{X>B}].
\]
For the lower tail, using the tail-integral formula,
\[
    \mathbb E[(-X)\mathbf 1_{X<0}]
    =
    \int_0^\infty \Pr[X<-t]\,dt.
\]
Since \(\mu\ge B/5\), for all \(t\ge 0\),
\[
    \Pr[X<-t]
    =
    \Pr[X-\mu<-(t+\mu)]
    \le
    \exp\left(-\frac{(t+B/5)^2}{2CB}\right).
\]
Therefore
\[
    \mathbb E[(-X)\mathbf 1_{X<0}]
    \le
    \int_0^\infty
    \exp\left(-\frac{(t+B/5)^2}{2CB}\right)\,dt.
\]
Using
\[
    (t+B/5)^2
    \ge
    (B/5)^2+\frac{2B}{5}t,
\]
we get
\[
    \mathbb E[(-X)\mathbf 1_{X<0}]
    \le
    \exp\left(-\frac{B}{50C}\right)
    \int_0^\infty \exp\left(-\frac{t}{5C}\right)\,dt
    =
    5C\exp\left(-\frac{B}{50C}\right).
\]

Similarly, since \(B-\mu\ge B/5\), the upper-tail integral yields
\[
    \mathbb E[(X-B)\mathbf 1_{X>B}]
    \le
    5C\exp\left(-\frac{B}{50C}\right)
\]
by the same calculation as above. Consequently,
\[
    \left|\mathbb E[\clip_B(X)]-\mu\right|
    \le
    10C\exp\left(-\frac{B}{50C}\right).
\]

Now choose \(\gamma=\frac{1}{100C}\). For all sufficiently large \(B\),
\[
    4\exp\left(-\frac{B}{50C}\right)
    \le
    e^{-\gamma B}
    \qquad\text{and}\qquad
    10C\exp\left(-\frac{B}{50C}\right)
    \le
    B e^{-\gamma B}.
\]
Thus
\[
    \Pr[\clip_B(X)\ne X]\le e^{-\gamma B},
    \text{ and }
    \left|\mathbb E[\clip_B(X)]-\mu\right|
    \le B e^{-\gamma B}.
\]
\end{proof}

\section{Proof of Lemma~\ref{lem:poisson-subg-coupling}}
\label{app:poisson-subg-coupling}

\poissonsubgcoupling*

\begin{proof}
Fix a sufficiently large absolute constant \(K>1\), for example \(K=4\), and
let
\[
    A=\{X\le K\mu\}.
\]
By the standard Chernoff bound for Poisson random variables,
\[
    \mathbb{P}(A^c)=\mathbb{P}(X>K\mu)\le e^{-c_K\mu}
\]
for some constant \(c_K>0\).

Define
\[
    m
    :=
    \frac{\mathbb{E}\left[X\mathbf{1}_{A^c}\right]}
         {\mathbb{P}(A^c)}
    =
    \mathbb{E}[X\mid A^c],
\]
when \(\mathbb{P}(A^c)>0\). Standard Poisson tail estimates imply
\[
    m \le C_K \mu
\]
for an absolute constant \(C_K\).

First define a possibly non-integer random variable
\[
    G
    =
    \begin{cases}
        X, & X\in A,\\
        m, & X\notin A.
    \end{cases}
\]
Then
\[
    \mathbb{E}[G]
    =
    \mathbb{E}[X\mathbf{1}_{A}]
    +
    m\mathbb{P}(A^c)
    =
    \mathbb{E}[X\mathbf{1}_{A}]
    +
    \mathbb{E}[X\mathbf{1}_{A^c}]
    =
    \mathbb{E}[X]
    =
    \mu.
\]
Also,
\[
    \mathbb{P}(G=X)\ge \mathbb{P}(A)\ge 1-e^{-c_K\mu}.
\]

It remains to note that \(G-\mu\) is sub-Gaussian with variance
\(O(\mu)\). Indeed, \(G\) agrees with \(X\) except on the event \(A^c\), whose
probability is \(e^{-\Omega(\mu)}\), and \(G\) is supported on an interval of
length \(O(\mu)\). Furthermore, for deviations \(0\le t\le O(\mu)\), the tails
of \(G-\mu\) are bounded by the corresponding Poisson Chernoff tails, up to the
additional exceptional mass \(e^{-\Omega(\mu)}\). Therefore
\[
    \mathbb{P}(|G-\mu|\ge t)
    \le
    2\exp\left(-\Omega\left(\frac{t^2}{\mu}\right)\right)
    \qquad \text{for all } t\ge 0.
\]
By Lemma~\ref{lem:subg-tail-mgf}, this implies
\[
    G-\mu \in \mathrm{SG}(C\mu)
\]
for an absolute constant \(C>0\).

Finally, to make \(G\) integer-valued, replace the value \(m\) on \(A^c\) by an
independent two-point random variable \(Z\in\{\lfloor m\rfloor,\lceil m\rceil\}\)
satisfying
\[
    \mathbb{E}[Z]=m.
\]
Using \(Z\) instead of \(m\) preserves the identity \(\mathbb{E}[G]=\mu\), preserves
the coupling probability, and only changes constants in the sub-Gaussian proxy.
\end{proof}

\section{Proof of Lemma~\ref{lem:subg-binomial-thinning}}
\label{app:subg-binomial-thinning}

\subgbinomialthinning*

\begin{proof}
First,
\[
    \mathbb{E}[H]
    =
    \mathbb{E}\left[\mathbb{E}[H\mid G]\right]
    =
    \mathbb{E}\left[\frac{G}{3}\right]
    =
    \frac{\mu}{3}.
\]

Fix \(\lambda\in\mathbb{R}\). We decompose
\[
    H-\frac{\mu}{3}
    =
    \left(H-\frac{G}{3}\right)
    +
    \frac{G-\mu}{3}.
\]
Conditioning on \(G\), and applying Lemma~\ref{lem:binomial-subg}, gives
\[
    \mathbb{E}\left[
        \exp\left(\lambda\left(H-\frac{G}{3}\right)\right)
        \,\middle|\, G
    \right]
    \le
    \exp\left(\frac{\lambda^2 G}{8}\right)
    \le
    \exp\left(\frac{\lambda^2 B}{8}\right).
\]
Therefore,
\[
\begin{aligned}
    \mathbb{E}\left[
        \exp\left(\lambda\left(H-\frac{\mu}{3}\right)\right)
    \right]
    &=
    \mathbb{E}\left[
        \exp\left(\frac{\lambda(G-\mu)}{3}\right)
        \mathbb{E}\left[
            \exp\left(\lambda\left(H-\frac{G}{3}\right)\right)
            \,\middle|\, G
        \right]
    \right] \\
    &\le
    \exp\left(\frac{\lambda^2 B}{8}\right)
    \mathbb{E}\left[
        \exp\left(\frac{\lambda(G-\mu)}{3}\right)
    \right].
\end{aligned}
\]
Since \(G\in \mathrm{SG}(\mu,v)\), Lemma~\ref{lem:subg-scaling} implies
\(\frac{G}{3}\in \mathrm{SG}\left(\frac{\mu}{3},\frac{v}{9}\right)\),
and hence
\[
    \mathbb{E}\left[
        \exp\left(\frac{\lambda(G-\mu)}{3}\right)
    \right]
    \le
    \exp\left(\frac{\lambda^2 v}{18}\right).
\]
Combining the two bounds,
\[
    \mathbb{E}\left[
        \exp\left(\lambda\left(H-\frac{\mu}{3}\right)\right)
    \right]
    \le
    \exp\left(
        \frac{\lambda^2}{2}
        \left(\frac{B}{4}+\frac{v}{9}\right)
    \right).
\]
Thus \(H\in \mathrm{SG}\left(\frac{\mu}{3},\frac{v}{9}+\frac{B}{4}\right)\),
as claimed.
\end{proof}

\end{document}